\documentclass[twoside]{article}
\usepackage[utf8x]{inputenc}
\usepackage[T1]{fontenc} 
\usepackage{RR}

\newif\ifdraft
\draftfalse

\usepackage{cleveref}
\RRNo{8017} 
\RRdate{July 2012}

\usepackage{listings}
\usepackage{xparse}
\usepackage[english]{babel}
\usepackage{amssymb} 

\usepackage{xspace}
\usepackage{empheq}
\usepackage{stmaryrd}       
\usepackage{multicol}
\usepackage{mathtools}      
\usepackage[vlined,algoruled,linesnumbered]{algorithm2e} 
\usepackage[table]{xcolor}
\usepackage{longtable} 
\usepackage{booktabs}      

\usepackage{graphicx} 
\usepackage{comment}

\usepackage{eurosym}

\usepackage{subfig}
\usepackage{wrapfig}
\usepackage{framed}

\usepackage{mdwlist}
\usepackage{enumitem}

\usepackage{amsthm}

\newcommand{\Var}[1]{\ensuremath{\text{Var}(#1)}}

\DeclareMathOperator*{\openc}{\ensuremath{\lbrace\vert\_\vert\rbrace_\_}}

\DeclareMathOperator*{\opaenc}{\ensuremath{\lbrace\_\rbrace_\_}}
\DeclareMathOperator*{\oppriv}{\ensuremath{inv}}
\DeclareMathOperator*{\oppk}{\ensuremath{pk}}
\DeclareMathOperator*{\oppair}{\ensuremath{pair}}

\DeclareMathOperator*{\ovars}{\ensuremath{Vars}}

\DeclareMathOperator*{\odom}{\ensuremath{dom}}
\DeclareMathOperator*{\oimg}{\ensuremath{img}}

\newcommand*{\pair}[2]{\ensuremath{\oppair\left(#1,#2\right)}}
\newcommand*{\enc}[2]{\ensuremath{\lbrace\vert #1 \vert\rbrace_{#2}}}

\newcommand*{\aenc}[2]{\ensuremath{\left\lbrace#1\right\rbrace_{#2}}}
\newcommand*{\sig}[2]{\ensuremath{\left\lbrace#1\right\rbrace_{#2}^{\text{sig}}}}
\DeclareMathOperator*{\opsig}{\ensuremath{\sig\_\_}}

\newcommand*{\priv}[1]{\ensuremath{\oppriv\left(#1\right)}}
\newcommand*{\pk}[1]{\ensuremath{\oppk\left(#1\right)}}

\newcommand*{\vars}[1]{\ensuremath{\ovars\left(#1\right)}}

\newcommand*{\dom}[1]{\ensuremath{\odom\left(#1\right)}}
\newcommand*{\img}[1]{\ensuremath{\oimg\left(#1\right)}}

\newcommand*{\card}[1]{\ensuremath{\lvert#1\rvert}}

\newcommand*{\snd}[2]{\ensuremath{!_{#1} #2}}
\newcommand*{\rcv}[2]{\ensuremath{?_{#1} #2}}

\newcommand*{\channel}[3][]{\ensuremath{{#2} \ifthenelse{\equal{#1}{}}{\rightharpoonup}{\stackrel{#1}{\rightharpoonup}} {#3}}}

\newcommand*{\Universe}{\ensuremath{\mathcal{T}}}
\newcommand*{\UniverseG}{\ensuremath{\mathcal{T}_g}}
\newcommand*{\ConstrSys}[1]{\ensuremath{\mathcal{#1}}}
\newcommand*{\UniAt}[1]{\ensuremath{\mathcal{#1}}}
\newcommand*{\UniVar}[1]{\ensuremath{\mathcal{#1}}}

\newcommand*{\set}[1]{\ensuremath{\left\lbrace #1 \right\rbrace}}

\newcommand*{\call}[1]{\ensuremath{\mathcal{#1}}}

\newcommand{\Variables}{\call{X}\xspace}

\makeatletter
\def\pairesymbol#1#2{\ensuremath{ #1.#2}}
\def\paireaux #1.#2;{%
\ifx @#2@
#1%
\else
\expandafter\pairesymbol\expandafter{#1}{\paireaux #2;}%
\fi}
\makeatother
\newcommand\paire[1]{\paireaux #1.;}

\DeclareMathOperator*{\parent}{\ensuremath{rel}}

\definecolor{aslanBG}{rgb}{0.933, 1, 0.94} 
\definecolor{xmlBG}{rgb}{ 1, 0.94, 0.933} 

\lstdefinelanguage{ASLan}{morekeywords={=>,section~typeSymbol,section,signature,types,equations,inits,rules:,goals,intruder,hornClauses,initial_state,:=,hc,:=,:-,equal,leq,not,=[exists,],goal,attack_state,and,or,implies,forall,exists,X,Y,F,O, G,H,U,R,S,>,->,*,step},morecomment=[l]{\%}, breaklines=true, breakatwhitespace=true,  basicstyle=\small\normalfont\ttfamily, backgroundcolor=\color{aslanBG}}
\lstdefinelanguage{plain}{morecomment=[l]{\%}, breaklines=true, breakatwhitespace=true, escapeinside={`}{'}, basicstyle=\small\normalfont\ttfamily}

\lstnewenvironment{ASLan}
{\lstset{language=ASLan}}
{\lstset{language=plain}}

\lstnewenvironment{XML}
{\lstset{language=XML, backgroundcolor=\color{xmlBG}, morestring=[s]{>}{<}, showstringspaces=false, stringstyle=\it }}
{\lstset{language=plain}}

 \definecolor{shadecolor}{rgb}{0,0,0}

\newcommand{\fs}{\UniAt{F}\xspace}

\newcommand{\vs}{\UniVar{X}\xspace}
\newcommand{\as}{\UniVar{A}\xspace}

\newcommand{\css}{\ConstrSys{S}\xspace}

\DeclareMathOperator*{\opmgu}{mgu}
\newcommand*{\mgu}[1]{\opmgu\left(#1\right)}

\DeclareMathOperator*{\opcoeur}{\heartsuit}
\NewDocumentCommand{\coeur}{O{}m}{\opcoeur_{#1}\left(#2\right)}

\newcommand*{\smallpar}[1]{\textsc{#1}}

\newcommand*{\lab}[1]{\ensuremath{\coeur{#1}}}
\DeclareMathOperator*{\opRHSs}{R}
\newcommand*{\RHSset}[2]{\ensuremath{\opRHSs_{#1}\left(\ifthenelse{\equal{#2}{}}{\card{#1}}{#2}\right)}}
\newcommand*{\labelset}[2]{\ensuremath{\lab{#1}\ifthenelse{\equal{#2}{}}{}{\left[:#2-1\right]}}}

\def\Sub#1{\ensuremath{\text{\rm Sub}(#1)}}
\def\Next#1#2{\ensuremath{\text{\rm Next}_{#1}(#2)}}
\def\Var#1{\ensuremath{\text{\rm Var}(#1)}}

\newtheorem{corollary}{Corollary}
\newtheorem{theorem}{Theorem}
\newtheorem{lemma}{Lemma}
\newtheorem{definition}{Definition}[section]

\NewDocumentCommand{\hsend}{om}{\ensuremath{\rcv{}{#2}}}
\NewDocumentCommand{\anyrule}{o}{\ensuremath{\mathrel{\rlap{\hspace{.45em}$*$}\rightarrow}}}
\NewDocumentCommand\stdrule{o}{\ensuremath{\rightarrow}}
\NewDocumentCommand\hnoncereceive{om}{\ensuremath{\stdrule #2}}
\NewDocumentCommand{\possiblyequal}{}{\ensuremath{=_{?}}}

\NewDocumentCommand{\repl}{mm}{\ensuremath{#1|_{#2}}}
\NewDocumentCommand{\position}{mm}{\ensuremath{#1[{#2}]}}
\NewDocumentCommand{\termheight}{m}{\ensuremath{\operatorname{ht}\left(#1\right)}}

\newcommand*{\forbid}[2]{\ensuremath{\natural_{#1}#2}}

\newcommand*{\scss}{\Sub{{\css}}}

\NewDocumentCommand{\pscss}{O{T}}{\ensuremath{\Sub{#1}}}
\NewDocumentCommand{\byDec}{O{}O{}}{\ensuremath{\operatorname{byDec}^{#1}_{#2}}}
\NewDocumentCommand{\ruleDec}{O{}O{}m}{\ensuremath{\operatorname{\rho}^{#1}_{#2}\left(#3\right)}}
\newcommand{\MathFunction}[3]{\ensuremath{\text{\rm #2}_{#1}(#3)}}
\newcommand{\previous}[1]{\MathFunction{\css}{prev}{#1}}
\newcommand*{\Der}[1]{\MathFunction{}{Der}{#1}}

\newcommand{\In}[1]{\MathFunction{}{In}{#1}}
\newcommand{\Out}[1]{\MathFunction{}{Out}{#1}}
\newcommand{\condset}[2]{\set{#1\,:\,#2}}
\newcommand\Nonces{\ensuremath{\UniAt{C}_{\text{\rm med}}}}

\NewDocumentCommand{\cmnt}{O{T}m}{\marginpar{\footnotesize  #1: #2}}

\newcommand{\unif}{\ensuremath{=_{?}}}

\Crefname{proposition}{Proposition}{Propositions}
\crefname{proposition}{proposition}{propositions}

\Crefname{corollary}{Corollary}{Corollaries}
\crefname{corollary}{corollary}{corollaries}

\Crefname{definition}{Definition}{Definitions}
\crefname{definition}{definition}{definitions}

\Crefname{algocf}{Algorithm}{Algorithms}
\crefname{algocf}{algorithm}{algorithms}

\excludecomment{removable}

\RRtitle{Contraintes de deducibilité avec négation}
\RRetitle{Intruder deducibility constraints with negation.
Decidability and application to secured service compositions.\thanks{This work is  supported by FP7 AVANTSSAR \cite{avantssarproj} and FP7 NESSoS \cite{nessosproj} projects.}
}
\titlehead{Intruder deducibility constraints with negation.}

 \RRauthor{
    Tigran Avanesov
 \thanks[inria]{INRIA Nancy Grand Est,   France.  Email: {\{rusi, turuani\}@inria.fr}}
 \thanks[irit]{IRIT, Universit{\'e} de Toulouse,  France.  Email: {ychevali@irit.fr}}
 \thanks[snt]{SnT, Universit{\'e} du Luxembourg, Luxembourg.  Email: {tigran.avanesov@uni.lu}}
 \and 
  Yannick Chevalier
   \thanksref{irit}  
  \and 
    Micha{\"e}l Rusinowitch
    \thanksref{inria}
     \and  
   Mathieu Turuani     
\thanksref{inria}
    }
\authorhead{T.Avanesov, Y.Chevalier, M.Rusinowitch, M.Turuani}

\RRabstract{
The problem of finding a mediator  to compose secured services has been  
reduced in our former work to the problem 
of solving deducibility constraints similar 
to those  employed for cryptographic protocol analysis. 
We extend in this paper the mediator synthesis procedure 
by a construction for expressing that some data is not accessible  
to the mediator. Then we give  a decision  procedure 
for verifying that a mediator satisfying this non-disclosure policy
can be effectively synthesized.  This procedure has been implemented
in CL-AtSe, our protocol analysis tool.
The procedure extends  constraint solving  for cryptographic protocol analysis in a significative 
way as it is able to handle negative deducibility constraints without restriction. 
In particular it applies to all subterm convergent theories and therefore 
covers several interesting theories in formal security analysis 
including encryption, hashing, signature and pairing.}

\RRresume{Voir ``Abstract''}
\RRkeyword{Web services, orchestration, security policy, separation of duty, deducibility constraints,  cryptographic protocols, formal methods, tool}
\RRmotcle{Services Web, orchestration, politique de sécurité, séparation des tâches, contraintes de deducibilité, protocoles de sécurité, méthodes formelles, util}
\RRprojets{CASSIS}
\RCNancy

\begin{document}

\makeRR

\section{Introduction}

\subsection{Context}
Trust and security management in distributed frameworks is known to be
a non-trivial critical issue. It is particularly challenging in
Service Oriented Architecture where services can be discovered and
composed in a dynamic way.  Implemented solutions should meet the
seemingly antinomic goals of openness and flexibility on one hand and
compliance with data privacy and other regulations on the other hand.
We have demonstrated in previous
works~\cite{SCC-WSCA08,DBLP:conf/fmco/ChevalierMR10,DBLP:conf/esorics/AvanesovCMR11}
that functional agility can be achieved for services with a
message-level security policy by providing an automated service
synthesis algorithm. It resolves a system of deducibility constraints
by synthesizing a \emph{mediator} that may adapt, compose and analyze
messages exchanged between client services and having the
functionalities specified by a goal service. It is complete as long as
the security policies only apply to the participants in the
orchestration and not on the synthesized service nor on who is able to
participate. However security policies often include such
\emph{non-deducibility} constraints on the mediator.  For instance an
organisation may not be trusted to efficiently protect the customer's
data against attackers even though it is well-meaning.  In this case a
client would require that the mediator synthesized to interact with
this organization must not have direct access to her private data,
which is an effective protection even in case of total compromise.
Also it is not possible to specify that the mediator enforces
\textit{e.g.}  dynamic separation of duty, \textit{i.e.}, restrictions on
the possible participants based on the messages exchanged. 

Since checking whether a solution computed by our previous algorithm 
satisfies the non-deducibility constraints is not complete, 
we propose in this paper to solve during the automated synthesis of the mediator
both deducibility and non-deducibility constraints. The
former are employed to specify a mediator that satisfies the
functional requirements and the security policy on the messages
exchanged by the participants whereas the latter are employed to
enforce a security policy on the mediator and the participants to the
orchestration.

\paragraph{Original contribution.}
We have previously proposed decision
procedures~\cite{SCC-WSCA08,DBLP:conf/fmco/ChevalierMR10,DBLP:conf/esorics/AvanesovCMR11}
for generating a mediator from a high-level specification with
deducibility constraints of a goal service.  In this paper we extend
the formalism to include non-deducibility constraints in the
specification of the mediator and provide a decision procedure
synthesizing a mediator for the resulting constraint systems.
 
\paragraph{Related works.}
In order to understand and anticipate potential flaws in complex
composition scenarios, several approaches have been proposed for the
formal specification and analysis of secure
services~\cite{avantssar-tacas-court,CostaDM11}.  Among the works
dedicated to trust in multi-agent systems, the models closest to ours
are~\cite{HerzigLogicTrust,lorini} in which one can express that an
agent trusts another agent in doing or forbearing of doing an action
that leads to some goal.
To our knowledge no work has previously considered the automatic
orchestration of security services with policies altogether as ours.
However there are some interesting related attempts to analyze
security protocols and trust management~\cite{Martinelli05,FrauD11}.
In ~\cite{Martinelli05} the author uniformly models security protocols 
and  access control based on trust management. 
The work  introduces an elegant approach to model automated trust negotiation. 
We also  consider an integrated framework for protocols  and  policies but  in our case 
$i)$ policies can  be explicitly  negative such as non-disclosure policies 
and separation-of-duty $ii)$ we propose  a decision procedure for 
the related trust negotiation problem $iii)$ we do not consider 
indistinguishability properties. In ~\cite{FrauD11} 
security protocols are combined with authorization logics that can be expressed with acyclic Horn clauses.
The authors encode the derivation of authorization predicates (for a service) 
as  subprotocols and  can reuse in that way the constraint solving algorithm  from 
\cite{MillenShmatikov2001} to obtain  a decision procedure. 
In our case we consider more general  intruder theories (subterm convergent ones) but focus on negation. 
We conjecture that our approach applies to their authorization policies too.

Our decision procedure for general (negative and positive) constraints
extend ~\cite{CorinES06} where negative constraints are limited to
have ground terms in right-hand sides, and the deduction system is
Dolev-Yao system ~\cite{DY83}, a special instance of the subterm deduction systems
we consider here.  In~\cite{kusters} the authors study a class of
contract signing protocols where some very specific Dolev-Yao negative
constraints are implicitly handled.

Finally one should note that the non deducibility 
constraints we consider tell that some data cannot be disclosed \emph{globally} 
but they cannot express finer-grained privacy or information leakage 
notions relying on probability such as for instance differential privacy. 

\paragraph{Paper organization.} 
In Subsection \ref{subsec:loan} we introduce a motivating banking application 
and sketch  our approach to obtain a mediator service. To our knowledge 
this application is out of the scope of alternative automatic methods. 
In Section  \ref{sec:constraint} we present our formal setting.
A deduction system  (Subsection \ref{subsec:deduction}) describes the abilities of the mediator to
process the messages. 
The mediator synthesis problem is reduced to the 
resolution of constraints that are defined in Subsection \ref{sec:constraint}.
In Section  \ref{sec:subterm} we recall the class of \emph{subterm deduction systems} 
and their properties. These systems have nice properties that allow us 
to decide in Section \ref{sec:deciding} the satisfiability of 
deducibility constraints even with negation.  
Finally we conclude in Section~\ref{sec:conclusion}.

\subsection{Synthesis of a Loan Origination Process (LOP)}
\label{subsec:loan}
We illustrate how negative constraints are needed to express
elaborated policies such as Separation of Duty by a classical loan
origination process example.  Our goal is to synthesize a mediator
that selects two bank clerks satisfying the Separation of Duty policy
to manage the client request.  Such a problem is solved automatically
by the decision procedure proved in the following sections. Let us
walk through the specification of the different parts of the
orchestration problem.

\paragraph{Formal setting.} Data are represented by first-order terms
defined on a signature that comprises binary symbols for symmetric and
assymetric encryptions (resp.
\ensuremath{\openc}, \ensuremath{\opaenc}), signature
(\ensuremath{\opsig}), and pairing (\ensuremath{\oppair}). Given a
public key $k$ we write \priv{k} its associated private key.  For
example \sig{a}{\priv{k}} is the signature of $a$ by the owner of
public key $k$.  For readability we write $\paire{a.b.c}$ a term
\pair{a}{\pair{b}{c}}.  The binary symbol $\parent$ expresses that
two agents are related and is used for defining a Separation of
Duty policy. A unary symbol $g$ is employed to designate participants
identity in the ``relatives'' database.

\paragraph{Client and  clerks.}
The client and the clerks are specified by services with a security
policy, specifying the cryptographic protections and the data and
security tokens, and a business logic that specify the sequence in
which the operations may be invoked. These are compiled into a sequence of protected messages each
service is willing to follow during the orchestration (Fig.~\ref{fig:lop:clerk} and~\ref{fig:lop:client}).

Client $C$  wants to ask for a loan from a service $P$, but 
for this he needs to get an approval from two banking clerks.
He declares his intention by sending to mediator $M$ a signed by him message 
containing service name $P$ and the identity of the client $g(C)$.
The mediator should send back the names of two clerks $A$ and $B$ who will evaluate his request.
The client then sends to each clerk a request containing amount $Amnt$, his name $C$ and a fresh key $N_k$ 
which should be used to encrypt decisions. 
Each request is encrypted with a public key of the corresponding  clerk ($pk(A)$ or $pk(B)$).
Then the mediator must furnish the decisions ($R_a$ and $R_b$) of two clerks each encrypted with the proposed key $N_k$ and also their signatures.
Finally, the client uses these tokens to ask his loan from $P$, where $pk(P)$ is a public key of $P$.

Clerk $A$ receives a request to participate in a LOP which is conducted by mediator $M$.
If he accepts, he returns his identity and public key. 
Then Clerk receives the client's request for a loan to evaluate: 
amount Amnt, client's name $C$ and a temporary key $K$ for encrypting his decision.
The last is sent back together with a signature certifying the authenticity of this decision
on the given request.

The client's non-disclosure policy is given in
Fig.~\ref{fig:lop:client} and is self-explanatory.  Let us explain the
services' non-disclosure policy.  The Clerk's decision (its last
message) should be unforgeable, thus, it should not be known by the
Mediator before it was sent by the Clerk (first non-disclosure
constraint of Fig.~\ref{fig:lop:clerk}).  The role clerk played by $A$
can be used by the mediator only if the constraint $\forbid{}{g(A)}$
is satisfied, showing that $A$ is not a relative with any other actor
of the protocol, as client and the other clerk (second non-disclosure
constraint of Fig.~\ref{fig:lop:clerk}).

\paragraph{Goal service.} In contrast with the
other services and clients, the goal service is only described in
terms of possible operations and available initial data.
\begin{description}
\item \emph{Initial data.}  Beside his private/public keys and the
  public keys of potential partners (\textit{e.g.} \pk{P}) the goal service has
  access to a relational database $\parent(g(a),g(c)), \parent(g(b),g(c)), \ldots$
  for storing known existing relations between agents to be checked
  against conflict of interests.
\item \emph{Deduction rules.}  The access to the database as well
  as the possible operations on messages are modeled by a set of
  deduction rules (formally defined later). We anticipate on the rest
  of this paper, and present the rules specific to this case study
  grouped into composition and decomposition rules in
  Fig.~\ref{fig:lop:deduction}.
\end{description}

\begin{figure}[tbp]
  \begin{minipage}[t]{.45\linewidth}
    \noindent\textbf{Clerk's ($A$) communications:}\footnotemark\par
    \begin{tabular}{r@{\ensuremath{~\Rightarrow~}}l@{~:~}p{0.7\linewidth}}
      $*$ & A & request.M\\
      A & M & \paire{g(A).\pk{A}}\\
      M & A &\aenc{\paire{\ensuremath{\text{\sl Amnt}}.C.K}}{\pk A} \\
      A & M & $m_1(A,Resp_A,K,C,\ensuremath{\text{\sl Amnt}})$ \\
    \end{tabular}\par
    \vspace*{1ex}
    
    \noindent\textbf{Non-disclosure constraints:}
    
    \begin{enumerate}
    \item $M$ cannot deduce the fourth message before it is sent by $A$.
    \item $M$ cannot deduce $g(A)$ before the second message is sent by $A$.
    \end{enumerate}

    \caption{Clerk's communications and non-disclosure constraints}
    \label{fig:lop:clerk}
  \end{minipage}\qquad%
  \begin{minipage}[t]{.48\linewidth}
    \noindent\textbf{Client's ($C$) communications:}\footnotemark[\value{footnote}]\par
    \begin{tabular}{r@{\ensuremath{~\Rightarrow~}}l@{~:~}p{0.7\linewidth}}
      C & M & \sig{\paire{g(C).loan.P}}{\priv{pk(C)}}\\
      M & C & \paire{A.B}\\
      C & M &
       $\paire{m_2(A,\ensuremath{\text{\sl Amnt}}).m_2(B,\ensuremath{\text{\sl Amnt}})}$\\
      M & C & \paire{m_3(A,R_a).m_3(B,R_b)}\\
      C & P & $m_4(\pk{P},A,B,R_a,R_b)$\\
    \end{tabular}\par
    \vspace*{1ex}

    \noindent\textbf{Non-disclosure constraints:}
    
    \begin{enumerate}
    \item M cannot deduce the amount $Amnt$.
    \item M cannot deduce $A$'s decision $R_a$.
    \item M cannot deduce $B$'s decision $R_b$.
    \end{enumerate}

    \caption{Client's Communications and non-disclosure constraints}
    \label{fig:lop:client}
  \end{minipage}
  
  \vspace*{1em}

  \begin{minipage}{1.0\linewidth}
    \centering
    \begin{tabular}{rcl|rcl@{\hspace*{1em}}rcl}
      \multicolumn{3}{c|}{\bf Composition rules}  & \multicolumn{6}{c}{\bf Decomposition rules}  \\
      \hline
      $x, y$ &$\to $&${\pair{x}{y}}  $&   $     \pair{x}{y} $&$\to $&$ x $  &&&\\
      &     &              &         $\pair{x}{y}$ &$\to$ & $y$ & $x,\parent(x,y)$ &$\to$ & $y$\\
      $x, y$ &$\to$ &${\enc{x}{y}}  $ &   $ y, \enc{x}{y} $ &$\to $& $x$ &  $y,\parent(x,y)$ &$\to$ & $x$\\
      $x, y$ &$\to$ &${\aenc{x}{y}}$   &   $\priv y, \aenc{x}{y} $ &$\to $& $x$&&&\\
      $x, \priv{y}$  & $\to$ & ${\sig{x}{\priv{y}}}$   & $y,{\sig{x}{\priv{y}}}$ &$\to $& $x$ &&&\\
    \end{tabular}
    \caption{Deduction system for the LOP example.}
    \label{fig:lop:deduction}
  \end{minipage}
\end{figure}

\paragraph{Mediator synthesis problem.}
In order to communicate with the services (here the client, the clerks and the service $P$), a mediator has to 
satisfy a sequence of constraints expressing that \textit{(i)} each message $m$  expected by a service (denoted $\rcv{}{m}$) 
can be deduced from all the previously sent messages $m'$ (denoted $\snd{}{m'}$) and the initial knowledge and 
\textit{(ii)} each message $w$ 
that should not be known or disclosed (denoted  $\forbid{}{w}$ and called negative constraint) 
is not deducible. 

The orchestration problem consists in finding a satisfying  
interleaving of the  constraints imposed by each service. 
For instance, clerk's and client's constraints  extracted from 
Fig.~\ref{fig:lop:clerk} and~ Fig.~\ref{fig:lop:client} are:
$$
\left\lbrace
  \begin{array}{rl}
    Client(C)\stackrel\Delta=&\snd{M}{\sig{\paire{g(C).loan.P}}{\priv{K_C}}} \  \rcv{M}{\paire{A.B}} \  \snd{M}{\paire{m_2(A,\ensuremath{\text{\sl Amnt}}).m_2(B,\ensuremath{\text{\sl Amnt}})}} \\
    &\rcv{M}{\paire{m_3(A, R_a).m_3(B, R_b)}} \  \forbid{M}{\text{\sl Amnt}} \ \forbid{M}{R_A} \  \forbid{M}{R_B} \   \\
     & \snd{P}{m_4(\pk{P},A,B,R_a,R_b)} \\
    Clerk(A) \stackrel\Delta= & \rcv{}{request.M} \ \forbid{M}{g(A)} \  \snd{M}{\paire{g(A).\pk{A}}}  \ \rcv{M}{\aenc{\paire{\ensuremath{\text{\sl Amnt}}.C.K}}{\pk A}} \\
     & \forbid{M}{m_1(A,Resp_A,K,C,\ensuremath{\text{\sl Amnt}})} \ \snd{M}{m_1(A,Resp_A,K,C,\ensuremath{\text{\sl Amnt}})}\\
  \end{array}
\right.
$$

If it exists our procedure outputs a solution which can be translated automatically into a mediator. 
Note, for example, that without the negative constraint $\forbid{}{g(A)}$  
a synthesized mediator might accept any  clerk identity and 
that could violate the Separation of Duty policy.

 \footnotetext{We
  have employed the following abbreviations for messages:\par
  \centerline{\ensuremath{\scriptscriptstyle\left\lbrace
        \begin{array}{rcl}
          m_1(A,Resp,K,Ct,S) &=& \sig{\paire{h(\paire{ A . S . Ct . Resp })}}{\priv{pk(A)}}.\enc{Resp}{K} \\
          m_2(A,S) &= & \aenc{\paire{S.C.N_k}}{\pk{A}} \\
          m_3(A,R)& =&  m_1(A,R,N_k,C,\text{\sl Amnt})\\
          m_4(K_0,A,B,R_1,R_2)&=& \paire{\aenc{\paire{\ensuremath{\text{\sl Amnt}}.C.A.R_1.B.R_2}}{K_0}.m_3(A,R_1).m_3(B,R_2)}\\
        \end{array}
      \right.}}}

\section{Derivations and constraint systems}
\label{sec:constraint}
In our setting messages are terms generated or obtained
according to some elementary rules called \emph{deduction rules}.  A
\emph{derivation} is a sequence of deduction rules applied by a
mediator to build new messages. The goal of the synthesis is specified
by a \emph{constraint system}, \textit{i.e.} a sequence of terms
labelled by symbols \ensuremath{!},\ensuremath{?}  or
\ensuremath{\natural}, respectively \emph{sent}, \emph{received}, or
\emph{unknown} at some step of the process.

\subsection{Terms and substitutions}

Let $\vs$ be a set of \emph{variables}, $\fs$ be a set of
\emph{function
  symbols} 
and $\UniAt{C}$ a set of \emph{constants}. The set of \emph{terms}
$\Universe$ is the minimal set containing $\vs$, $\UniAt{C}$ and if
$t_1,\dots,t_k \in \Universe$ then $f(t_1,\dots,t_k)\in\Universe$ for
any $f\in\fs$ with arity $k$.
The set of \emph{subterms} of a term $t$ is denoted $\Sub{t}$ and is
the minimal set containing $t$ such that $f(t_1, \dots,
t_n)\in\Sub{t}$ implies $t_1, \dots,t_n\in\Sub{t}$ for $f\in\fs$.  We
denote \vars{t} the set $\vs\cap\Sub{t}$%
.  
A term $t$ is \emph{ground} is $\vars{t}=\emptyset$. We
denote $\UniverseG$ the set of ground terms.  

A \emph{substitution}
$\sigma$ is an idempotent mapping from $\vs$ to $\Universe$. It is ground if it is
a mapping from $\vs$ to $\UniverseG$.  The application of a
substitution $\sigma$ on a term $t$ is denoted $t\sigma$ and is equal
to the term $t$ where all variables $x$ have been replaced by the term
$x\sigma$.  We say that a substitution $\sigma$ is \emph{injective} on
a set of terms $T$, iff for all $p,q\in T$ $p\sigma=q\sigma$ implies
$p=q$.  The \emph{domain} of $\sigma$ (denoted by $\dom{\sigma}$) is
set: $\condset{x\in\vs}{x\sigma\neq x}$. The \emph{image} of $\sigma$
is $\img{\sigma}=\condset{x\sigma}{x\in\dom{\sigma}}$.  Given two
substitutions $\sigma,\delta$, the substitution $\sigma\delta$ has for
domain $\dom{\sigma}\cup\dom{\delta}$ and is defined by
$x\sigma\delta=(x\sigma)\delta$.  If
$\dom{\sigma}\cap\dom{\delta}=\emptyset$ we write $\sigma\cup\delta$
instead of $\sigma\delta$.

A \emph{unification system} $U$ is a finite set of equations $\set{p_i
  \possiblyequal q_i}_{1\le i\le n}$ where $p_i, q_i \in \Universe$. A
substitution $\sigma$ is an \emph{unifier} of $U$ or equivalently
satisfies $U$ iff for all $i=1,\dots,n$, $p_i\sigma=q_i\sigma$.  Any
satisfiable unification system $U$ admits a \emph{most general
  unifier} $\mgu{U}$, unique modulo variable renaming, and such that
for any unifier $\sigma$ of $U$ there exists a substitution $\tau$
such that $\sigma=\mgu{U}\tau$. Wlog we assume in the rest of this
paper that $\vars{\img{\mgu{U}}} \subseteq \vars{U}$, \textit{i.e.},
the most general unifier does not introduce new variables.

A \emph{sequence} $s$ is indexed by $[1,\ldots,n]$ with $n \in
\mathbb{N}$.  We write \card{s} the length of $s$, $\emptyset$ the
empty sequence, $s[i]$ the ith element of $s$, $s[m:n]$ the sequence
$s[m],\dots,s[n]$ and $s,s'$ the concatenation of two sequences $s$
and $s'$.  We write $e\in s$ and $E\subseteq s$ for, respectively,
$\exists i: s[i]=e$ and $\forall e\in E, e\in s$.

\subsection{Deduction systems}
\label{subsec:deduction}
The new values created by the mediator are constants in a subset
\Nonces\ of \UniAt{C}. We assume that both $\Nonces$ and
$\UniAt{C}\setminus \Nonces$ are infinite. Given $l_1,\ldots,l_n,r\in
\Universe$, the notation $l_1,\ldots,l_n\stdrule r$ denotes a
\emph{deduction rule} if $\Var{r}\subseteq \bigcup_{i=1}^n\Var{l_i}$.
A \emph{deduction} is a ground instance of a deduction rule.
A \emph{deduction system} is a set of deduction rules that contains a
finite set of deduction rules in addition to all \emph{nonce creation
  rules} $\hnoncereceive{n}$ (one for every $n \in \Nonces$) and all
reception rules $\hsend{t}$ (one for every $t\in\Universe$).
All rules but the reception rules are called \emph{standard} rules.
The deduction system describes the abilities of the mediator to
process the messages.  In the rest of this section we fix an arbitrary
deduction system $\mathcal{D}$.  We denote by $l\anyrule r$ any rule
and $l\stdrule r$ any standard rule.

\subsection{Derivations and localizations}
\label{subsec:derivation}


A \emph{derivation} is a sequence of deductions, including receptions
of messages from available services, performed by the mediator. Given
a sequence of deductions $E=(l_i\anyrule r_i)_{i=1,\dots,m}$ we denote
$\RHSset{E}{i}$ the set $\condset{r_j}{j\le i}$.
\begin{definition}[Derivation]
  A sequence of deductions $D=(l_i\anyrule r_i)_{i=1,\dots,m}$ is a
  \emph{derivation} if for any $i\in \set{1,\dots,m}$, $l_i\subseteq
  \RHSset{D}{i-1}$.
\end{definition}
Given a derivation $D$ we define $\Next{D}{i}=\min(\set{\card{D}+1}\cup \{j : j>i \mbox{ and } D[j] = \hsend{t_j}\})$.
The explicit knowledge of the mediator is the set of terms it has
already deduced, and its implicit knowledge is the set of terms it can
deduce.  If the former is $K$ we denote the latter $\Der{K}$. A
derivation $D$ is a \emph{proof} of $s\in\Der{K}$ if $\hsend{r}\in D$
implies $r\in K$, and $ D[\card{D}]= l \anyrule t$.  Thus, we have: \par
\centerline{\ensuremath{\Der{K}=\condset{t}{\exists  D \mbox{ derivation } \text{ s.t. } \hsend{r}\in D \text{
    implies } r\in K \text{, and } D[\card{D}]= l \stdrule t }}}
  

\subsection{Constraint systems}
\label{subsec:constraint}
\begin{definition}[Constraint system]
  A constraint system $\css$ is a sequence of constraints where each
  constraint has one of three forms (where $t$ is a term):

     \begin{enumerate}
     \item $\rcv{}{t}$, denoting a message reception by an available
       service or a client,
     \item $\snd{}{t}$, denoting a message emission by an available
       service or a client,
     \item $\forbid{}{t}$, a negative constraint, denoting that the
       mediator must not be able to deduce $t$ at this point;
     \end{enumerate} 
 
     and that satisfies the following properties for any $1\le i \le \card{\css}$:

     \begin{description}
     \item[Origination:] if
       $\css[i]=\snd{}{t_i}$ then $\vars{t_i}\subseteq
       \bigcup_{j<i}\vars{\condset{t_j}{ \css[j]=\rcv{}{t_j}}}$;
     \item[Determination:] if
       $\css[i]=\forbid{}{t_i}$ then $\vars{t_i}\subseteq
       \bigcup_{j}\vars{\condset{t_j}{ \css[j]=\rcv{}{t_j}}}$.
     \end{description}
 
\end{definition}
\emph{Origination} means that every unknown in a service's
state originates from previous input by the mediator.
\emph{Determination} means that negative constraints are on messages
determined by a service's state at the end of its execution.

In the rest of this paper \css\ (and decorations thereof) denotes a 
constraint system.  An index $i$ is a \emph{send} (\textit{resp.} a \emph{receive})
index if $\css[i]=\snd{}{t}$ (\textit{resp.} $\css[i]=\rcv{}{t}$) for
some term $t$. If $i_1,\ldots,i_k$ is the sequence of all send
(\textit{resp.} receive) indices in $\css$ we denote $\Out\css$
(\textit{resp.} $\In\css$) the sequence $\css[i_1],\ldots,\css[i_k]$.
We note that the origination and determination properties imply
$\Var{\css}=\Var{\In{\css}}$.  Given $1\le i \le \card{\css}$ we
denote $\previous{i}$ to be $\max (\set{0}\cup \condset{j}{j\le i \text{ and }
  \css[j]= \snd{}{t_j}})$.

\begin{definition}[Solution of a constraint system]\label{def:solution:sigma}
  A ground substitution $\sigma$ is a solution of $\css$, and we
  denote $\sigma\models\css$, if $\dom{\sigma}=\Var{\css}$ and
  \begin{enumerate}
  \item if $\css[i]= \rcv{}{t}$ then $t\sigma\in\Der{\condset{t_j\sigma}{j\le
      \previous{i} \text{ and } \css[j]=\snd{}{t_j}}}$ 
  \item if $\css[i]= \forbid{}{t}$ then $t\sigma\notin\Der{\condset{t_j\sigma}{j\le
      \previous{i} \text{ and } \css[j]=\snd{}{t_j}}}$ 
  \end{enumerate}
\end{definition}


\begin{definition}[Compliant derivations]
  Let $\sigma$ be a ground substitution with
  $\dom{\sigma}=\Var{\css}$.  A derivation $D$ is
  \emph{$(\css,\sigma)$-compliant} if there exists a strictly
  increasing bijective mapping $\alpha$ from the send indices of
  $\css$ to the set $\condset{j}{D[j] = \hsend{r}}$ such that
  $\css[i]=\snd{}{t}$ implies $D[\alpha(i)]=\hsend{t\sigma}$.
\end{definition}
An example of $(\css,\sigma)$-compliant derivation is shown in
\Cref{fig:css|der}. Since a sequence of receptions is a derivation, we
note that for every ground substitution $\sigma$ with
$\dom{\sigma}=\Var{\In{\css}}$ there exists at least one compliant
derivation $D$.

\ifpdf
\begin{figure}
\centering
 \includegraphics[width=.95\textwidth]{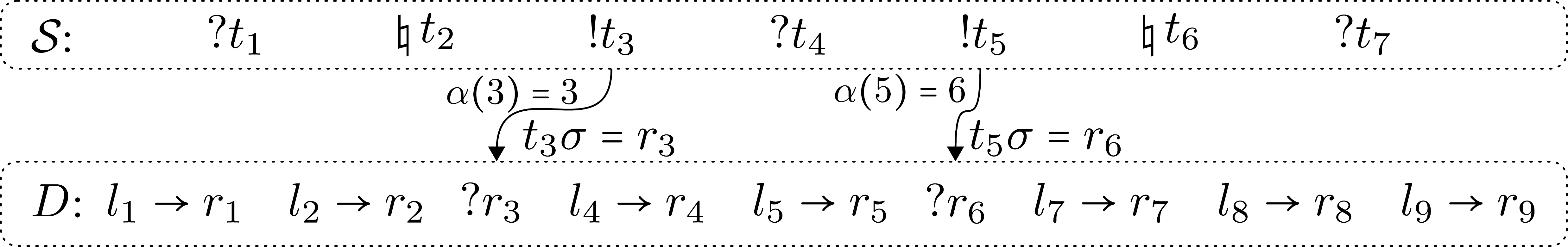}
 \caption{\label{fig:css|der} A constraint system and a compliant
   derivation} 
\end{figure}
\fi

\begin{definition}[Proof of a solution]
  Let $\sigma$ be a ground substitution. A derivation $D$ is a
  \emph{proof} of $\sigma\models\css$, and we denote
  $D,\sigma,\alpha\vdash \css$, if:
\begin{enumerate}
\item $D$ is $(\css,\sigma)$-compliant with the mapping $\alpha$ and 
\item if $\css[i]=\rcv{}{t}$ there is 
  $j<\Next{D}{\alpha(\previous{i})}$ such that $D[i] = l\anyrule t\sigma$ and 
\item if $\css[i]=\forbid{}{t}$ then $t\sigma\notin\Der{\condset{t_j\sigma}{j\le
      \previous{i} \text{ and } \css[j]=\snd{}{t_j}}}$.
\end{enumerate}
\end{definition}

In \Cref{fig:css|der}, if $\sigma$ is a solution of $\css$ and, for
example, $t_1\sigma = r_2$, $t_2\sigma\notin\Der{\emptyset}$,
$t_4\sigma=r_4$, $t_6\sigma\notin\Der{\set{r_3,r_6}}$ and $t_7\sigma =
r_8$ then $D$ is a proof of $\sigma \models \css$.

Let us prove that if $\sigma\models\css$ then there is a proof $D,\sigma,\alpha\vdash \css$.


\begin{definition}[Maximal derivation] \label{def:maximal:proof} Let
  $T$ be a finite set of terms and $\sigma$ be a ground substitution
  with $\dom{\sigma}=\Var{T}$. A derivation $D$ is
  $(T,\sigma)$-maximal iff for every $t\in\Sub{T}$,
  $t\sigma \in \Der{\RHSset{D}{i}}$ implies $t\sigma \in
  \RHSset{D}{\Next{D}{i}-1}$.
\end{definition}
First we prove that maximal derivations are natural proof candidates of $\sigma\models\css$.
\begin{lemma}\label{:model|criteria}
  Let $\sigma$ be a ground substitution with $\dom\sigma=\Var{\css}$
  and $D$ be a $(\css,\sigma)$-compliant $(\scss,\sigma)$-maximal
  derivation.  \textbf{Then } $\sigma \models \css$ iff for
  all $i$
  \begin{itemize*}
   \item if $\css[i]=\rcv{}{t}$ then there exists 
              $ j < \Next{D}{\alpha(\previous{i})}:  D[j]=l\anyrule t\sigma$ and 
   \item if $ \css[i]=\forbid{}{t}$ then 
          for all $j< \Next{D}{\alpha(\previous{i})}:  D[j] \neq l\anyrule t\sigma$. 
  \end{itemize*}
\end{lemma}

In the next lemma we show that any $(T,\sigma)$-maximal derivation $D$
may be extended into a $(T',\sigma')$-maximal derivation 
for an arbitrary extension $T',\sigma'$ of $T, \sigma$
by adding into $D$ only standard deductions.

\begin{lemma}{\label{lemma:existence:maximal:localization}}  	
  Let $\sigma$ be a ground
  substitution with $\dom{\sigma}=\Var{\css}$.
  Let $T_1,T_2$ be two sets of terms such that $T_1\subseteq T_2$, and
  $\sigma_1,\sigma_2$ be two substitutions such that $\dom{\sigma_1} =
  \Var{T_1}$ and $\dom{\sigma_2}=\Var{T_2}\setminus \Var{T_1}$.
  If $D$ is a $(T_1,\sigma_1)$-maximal $(\css,\sigma)$-compliant
  derivation in which no term is deduced twice by a standard rule, 
  \textbf{then}
  there exists a
  $(T_2,\sigma_1\cup\sigma_2)$-maximal
  $(\css,\sigma)$-compliant derivation $D'$ in which no term is deduced twice by a standard rule 
  such that every deduction
  whose right-hand side is in $\Sub{T_1}\sigma_1$ occurs in $D'$ iff it
  occurs in $D$.
\begin{proof}


  Let $i_1,\ldots,i_k$ be the indices of the non-standard rules in
  $D$, let  $D[i_j]=\hsend{t_{i_j}}$, and let  for $0\le j\le k$
  $D_j=D[i_j+1:i_{j+1}-1]$ with $i_0=0$ and
  $i_{k+1}=\card{D}+1$. 
  That is, $D=D_{0}, \snd{}{t_{i_1}}, D_{1},\snd{}{t_{i_2}}, D_{2}, \dots \snd{}{t_{i_k}}, D_{k}$.
  Noting that
  $\dom{\sigma_1}\cap\dom{\sigma_2} = \emptyset$ 
  let $\sigma'=\sigma_1\cup\sigma_2$.
     
  For each $t\in \Sub{T_2}$ such that
  $t\sigma'\in\Der{t_{i_1},\ldots,t_{i_k}}$ let $i_t$ be minimal such
  that $t\sigma'\in\Der{t_{i_1},\ldots,t_{i_t}}$, and let $E^0_t$ be a
  proof of this fact, and $E_t$ be a sequence of standard deductions
  obtained by removing every non-standard deduction from $E^0_t$.

  For $0\le j\le k$ let $D'_j$ be the sequence of standard deduction
  steps $D_j,E_{s_1},\ldots,E_{s_p}$ for all $s_m\in\Sub{T_2}\sigma'\setminus\Sub{T_1}\sigma'$ such that
  $i_{s_m}=j$ in which every rule of $E_{s_1},\ldots,E_{s_p}$ that deduces a term previously
  deduced in the sequence 
  or for some $m\le j$
   deduced in $D'_m$ or in $D[i_m]$  is removed.


  Let $D'=D'_0,\hsend{t_{i_1}},D'_1,\ldots,\hsend{t_{i_k}},D'_k$.  We
  have deleted in each $E_t^0$ only deductions whose right-hand side
  occurs before in $D'$, and thus $D'$ is a derivation. Since the
  $D'_i$ contains only standard deductions,
  we can see that $D'$ is 
  $(\css,\sigma)$-compliant. 

  Since $D$ is $(T_1,\sigma_1)$-maximal and no term is deduced twice
  in $D$ we note that, for $t\in T_1$, no standard deduction of
  $t\sigma_1$ from a sequence $D_j$ is deleted.			
  Furthermore we note that standard deductions of terms $T_2\sigma_2$
  that are also in $T_1\sigma_1$ are deleted by construction and by
  the maximality of $D$. Thus a deduction whose right-hand side is in
  $\Sub{T_1}\sigma_1$ is in $D'$ iff it occurs in $D$.

  By construction $D'$ is $(T_2,\sigma')$-maximal and no term is deduced twice by standard deductions.
\end{proof}

\end{lemma}

Taking $T_1=\emptyset$, $T_2=\Sub{{\css}}$, and $\sigma_2=\sigma$,
Lemma~\ref{lemma:existence:maximal:localization} implies that for
every substitution $\sigma$ of domain $\Var{\css}$ there exists a
$(\css,\sigma)$-compliant $(\Sub{{\css}},\sigma)$-maximal derivation
$D$.  By Lemma~\ref{:model|criteria} if $\sigma \models \css$ then $D$
is a proof of $\sigma \models \css$. Since the converse is trivial, it
suffices to search proofs maximal wrt $T\supseteq \Sub{\css}$.

\section{Subterm deduction system}
\label{sec:subterm}

\subsection{Definition and main property}

We say that a deduction system is a \emph{subterm deduction system}
whenever each deduction rule which is not a nonce creation or a
message reception is either:
\begin{enumerate}
\item $x_1,\ldots,x_n \stdrule f(x_1,\ldots,x_n)$ for a
  function symbol $f$;
\item $l_1,\ldots,l_n\stdrule r$ for some terms $l_1,\ldots,l_n, r$ such that  $r\in\bigcup_{i=1}^n \Sub{l_i}$. 
\end{enumerate}
A \emph{composition} rule is either a message reception, a nonce
creation, or a rule of the first type. A deduction rule is otherwise a
\emph{decomposition} rule.  Reachability problems 
for deduction
systems with a convergent equational theory 
are reducible to the
satisfiability of a constraint system in the empty theory for a
deduction system in our
setting~\cite{Lynch-Meadows-Arspa,Kourjieh-these}. If furthermore the
equational theory is \emph{subterm}~\cite{Baudet} the reduction is to
a subterm deduction system as just defined above.
 

Now we show that if $D,\sigma,\alpha \vdash \css$, a term
$s\in\Sub{D}$ is either the instance of a non-variable subterm of
${\Out{\css}}$ or deduced by a standard composition.

\begin{lemma}\label{lemma:variable|created|by|composition}
  Let $\sigma$ be a ground substitution such that $\sigma\models\css$.
  If $D$ is a proof of $\sigma\models\css$ such that no term is
  deduced twice in $D$ by standard rules and $s$ is a term such that
  $s\in\Sub{D}$ and $s\notin
  (\Sub{\Out{\css}}\setminus\Variables)\sigma$ \textbf{then} there
  exists an index $i$ in $D$ such that $D[i]=l\to s$ is a composition
  rule and $s\notin\Sub{\RHSset{D}{i-1}}$.
\end{lemma}

\begin{proof}
  First we note that by definition of subterm deduction systems for
  any decomposition rule $l \to r$ we have \textit{a)} $r\in\Sub{l}$,
  and \textit{b)} for any composition rule $l \to r$ we have $l
  \subset \Sub{r}$ and $\Sub{r}\setminus \Sub{l} = \set{r}$.

  Let $D$ be a proof of $\sigma\models\css$, and let $i$ be minimal
  such that $D[i]=l_r\anyrule r$ with $s\in\Sub{r}$.
  Since $l_r\subseteq \RHSset{D}{i-1}$, the minimality of $i$
  implies $s \in \Sub{r}\setminus \Sub{l_r}$.

  Thus by \textit{a)} $D[i]$ cannot be a decomposition.
  
  If $D[i]=\hsend{r}$ then by the $(\css,\sigma)$-compliance of $D$ we
  have ${\css}[\alpha^{-1}(i)] = \snd{}{t}$ with $t\sigma = r$.  We
  have $s\in\Sub{r} = \Sub{t\sigma}= \Sub{t}\sigma \cup
  \Sub{\vars{t}\sigma}$.  

  If $s\in (\Sub{\Out{S}}\setminus\Variables)\sigma$ we are done,
  otherwise there exists $y\in\vars{t}$ with $s
  \in\Sub{y\sigma}$.  By the origination property, there exists
  $k<\alpha^{-1}(i)$ such that $\css[k]=\rcv{}{t'}$ with
  $y\in\vars{t'}$.  Since $D, \sigma , \alpha \vdash \css$ and
  $k<\alpha^{-1}(i)$ there exists $j<i$ such that $D[j]=l_j\to
  t'\sigma$.  The minimality of $i$ is contradicted by
  $s\in\Sub{t'\sigma}$.

  Therefore, $D[i] = l_r \to r$ is a standard composition rule. 
  As a consequence, $\Sub{r}\setminus\Sub{l_r} = \set{r}$.
  Since $s \in \Sub{r}\setminus\Sub{l_r}$, we finally obtain $s = r$.
\end{proof}

\subsection{Locality}
\label{subsec:locality}
Subterm deduction systems are not necessarily local in the sense
of~\cite{locality}. However we prove in this subsection that given
$\sigma$, there exists a finite extension $T$ of \Sub{\css} and an
extension $\sigma'$ of $\sigma$ of domain $\Var{T}$ and a
$(T,\sigma')$-maximal derivation $D$ in which every deduction relevant
to the proof of $\sigma\models\css$ is liftable into a deduction
between terms in $T$.  Let us first precise the above statements.

\begin{definition}[Localization set]\label{def:localization}
  A set of terms $T$ \emph{localizes} a derivation $D = (l_i\anyrule
  r_i)_{1\le i\le m}$ for a substitution $\sigma$ of domain $\Var{T}$
  if for every $1\le i\le m$ if $D[i]$ is a standard rule and there exists $t\in
  \Sub{T}\setminus \Variables$ such that $t\sigma=r_i$, there exists
  $t_1,\ldots,t_n\in \Sub{T}$ such that
  $\set{t_1\sigma,\ldots,t_n\sigma}\subseteq \RHSset{D}{i-1}$ and
  $t_1,\ldots,t_n\stdrule t$ is the  instance of a
  standard deduction rule.
\end{definition}

First, we prove that for subterm deduction systems, every proof $D$ of
$\sigma \models \css$ is localized by a set $T$ of DAG size linear in the
DAG size of $\css$.

\begin{lemma}{\label{lemma:localization}}
  If $\sigma$ is a ground substitution such that $\sigma\models\css$
  there exists $T\supseteq \Sub{\css}$ of size linear in
  \card{\scss}, a substitution $\tau$ of domain
  $\Var{T}\setminus\Var{\css}$ and a $(T,\sigma\cup\tau)$-maximal and
  $(\css,\sigma)$-compliant derivation localized by $T$ for $\sigma\cup\tau$.
\end{lemma}

\begin{proof}

  By Lemma~\ref{lemma:existence:maximal:localization} applied with
  $T_1=\emptyset$, $T_2=\Sub{\css}$,
  $\sigma_1=\emptyset, \sigma_2=\sigma$, and $D_0$ the $(\css,\sigma)$-compliant derivation
  that has no standard deductions, there exists a
  $(\Sub{\css},\sigma)$-maximal
  $(\css,\sigma)$-compliant derivation $D$ in which no
  term is deduced twice by a standard deduction. From now on we
  let $T_0=\Sub{\css}$.

  Let $\set{l_i\to r_i}_{1\le i\le n}$ be the set of decompositions in
  $D$, and $\set{(L_i\to R_i,\tau_i)}_{1\le i\le n}$ be a set of
  decomposition rules and ground substitutions such that for all $1\le
  i\le n$ we have $L_i\tau_i\to R_i\tau_i=l_i\to r_i$. Since no term
  in $D$ is deduced twice by a standard deduction, by
  \Cref{lemma:variable|created|by|composition} we have $n\leq
  \card{\Sub{\Out{S}}}$.

   Modulo variable renaming we may assume that $i\neq j$ implies
   $\dom{\tau_i}\cap\dom{\tau_j}=\emptyset$, and thus that
   $\tau=\bigcup_{i=1}^n \tau_i$ is defined on $T_1 = \bigcup_{i=1}^n
   (\Sub{L_i}\cup\Sub{R_i})$. Note that the size of $T_1$ is
   bounded by $M\times \card{\Sub{\Out{\css}}}$, where $M$ is the
   maximal size of a decomposition rule belonging to the deduction
   system.

   Let $T=T_0\cup T_1$ and, noting that these substitutions are
   defined on non-intersecting domains, let
   $\sigma'=\sigma\cup\tau$. By construction
   $\card{T}\le (M+1)\times \card{\Sub{{\css}}}$.

  By Lemma~\ref{lemma:existence:maximal:localization} there exists a
  $(\css,\sigma)$-compliant derivation $D'$ which is
  $(T,\sigma')$-maximal and such that every deduction of a term in
  $T_0\sigma$ that occurs in $D$ also occurs in $D'$
  and no term is deduced twice in $D'$ by a standard deduction.
 
  Let $l\to r$ be a deduction in $D'$ which does not appear in $D$.
  Since $D$ is $(T_0,\sigma)$-maximal we have $r\notin
  \Sub{T_0}\sigma$, and thus $r\notin \Sub{\Out{\css}}\sigma$. Since
  no term is deduced twice in $D'$ by
  Lemma~\ref{lemma:variable|created|by|composition} this deduction
  must be a composition.

  Let us prove $D'$ is $(T,\sigma')$-localized. By definition of
  composition rules, every composition that deduces a term $t\sigma'$ with 
  $t\in\Sub{T}\setminus\Var{T}$ has a left-hand side $t_1\sigma',\dots,t_k\sigma'$ 
  with $t_1,\dots,t_k\in\Sub{T}$ and $t_1,\dots,t_k\to t$ is an instance of a composition rule.
  By the preceding paragraph every
  decomposition in $D'$ occurs in $D$ and thus by construction has its
  left-hand side in $T_1\sigma'$ which was previously built in $D$ 
  and is an instance of some $L_i \to R_i$ such that $\Sub{L_i\cup\set{R_i}}\subseteq T_1\subseteq T$.

  Thus every deduction whose right-hand side is in
  $(\Sub{T}\setminus\Var{T})\sigma'$ has its left-hand side in
  $\Sub{T}\sigma'$, and thus $D'$ is localized by $T$ for $\sigma'$.
\end{proof}


We prove now that to solve constraint systems one can first guess
equalities between terms in $T$ and then solve constraint systems
without variables.  The guess of equalities is correct wrt a solution
$\sigma$ if terms in $T$ that have the same instance by $\sigma$ are
syntactically equal.  We characterize these guesses as follows.

\begin{definition}[One-to-one localizations]\label{def:one:to:one}
  A set of terms $T$ \emph{one-to-one localizes} a derivation $D$ for
  a ground substitution $\sigma$ if $\sigma$ is injective on $\Sub{T}$
  and $T$ localizes $D$ for $\sigma$.
\end{definition} 

In Lemma~\ref{lemma:localization:one-to-one} we prove that once equalities between variables are correctly
guessed there exists a one-to-one localization of a maximal proof $D$.

\begin{lemma} \label{:mgu|is|a|unit} Let $T$ be a set of terms such
  that $T=\Sub{T}$, $\sigma$ be a ground substitution defined on
  $\vars{T}$, $U=\condset{p\possiblyequal q }{ p,q\in T \wedge
    p\sigma=q\sigma}$ be a unification system and $\theta$ be its most
  general idempotent unifier with $\vars{\img{\theta}}\subseteq
  \vars{U}$.  Then for any term $t$, $t\theta\sigma = t\sigma$.
\begin{proof}
 Let us show $\forall x\in\vars{T}, x\sigma = x\theta\sigma$. Note that this trivially holds if $x\theta = x$.
 Thus we consider case  $x\theta \neq x$.

 Since $U$ contains all equations ${p\possiblyequal p}$ for $p\in\Sub{T}=T$, we have  $\Sub{T}=\Sub{U}$. 
 From the idempotency of $\theta$  ($\forall y\in\vars{U}, y\theta\theta = y\theta$), we get
 $\forall y\in\vars{\img{\theta}}, y\theta = y$.  

 As $\sigma$ is evidently a unifier of $U$, there exists a substitution $\tau$ such that $\sigma= \theta\tau$
 Therefore, $y\sigma=y\theta\tau = y\tau$, i.e. 
 $y\sigma= y\tau$ for all  $y\in\vars{\img{\theta}}$.
 Thus, for any $x\in\vars{T}$, $x\theta \sigma = x\theta\tau = x\sigma$.
 
 Consequently, 
 for any term $t$ 
 we have $t\sigma=t\theta\sigma$.
\end{proof}

\end{lemma}

\begin{lemma}\label{:mgu|sub} 	
 Let $U$ be a unification system and $\theta=\mgu{U}$ an idempotent most general unifier with $\vars{\img{\theta}}\subseteq \vars{U}$. 
 Then $\forall p\in \Sub{\img{\theta}}$ $\exists q\in\Sub{U} : p=q\theta$.
\begin{proof}
 The case where $p\in\vs$ is trivial, since $\vars{\img{\sigma}}\subseteq \vars{U}$ and we can take $q=p$.
 Otherwise, suppose that $p\in\Sub{x\theta | x\in\dom{\theta}}\setminus \vs$ is such that $\forall q\in\Sub{U}$ $p\neq q\theta$.
 Let  $z$ be a fresh variable. 
 Let $\theta'= \condset{x\mapsto \repl{(x\theta)}{p\mapsfrom z}}{ x\in\dom{\theta}}$.
 Let us denote the height of a term $t$ by $\termheight{t}$, and a subterm of $t$ at position $l$ by $\position{t}{l}$ and 
 the set of all positions in $t$ by $\position{t}{}$.

 Let us prove that $\forall u,v\in\Sub{U}~~ u\theta=v\theta \implies u\theta'=v\theta'$.
\begin{itemize}
  \item If $u,v\in\vs$ then the statement is true by definition.
 \item If, w.l.o.g., $u\in\vs$ but $v \notin \vs$.
   Then for any $l\in\position{v}{}$ we have $\position{(u\theta)}{l}=(\position{v}{l})\theta$.
   Since $\position{v}{l}\in\Sub{U}$ we get $\position{(u\theta)}{l} \neq p$ (as we took such $p$ that  $\forall q\in\Sub{U}$ $p\neq q\theta$), and therefore, $(\position{v}{l})\theta\neq p$.
   Thus, $\repl{(v\theta)}{p\mapsfrom z}=v\theta'$.
    Therefore, $u\theta'=\repl{(u\theta)}{p\mapsfrom z} = \repl{(v\theta)}{p\mapsfrom z} = v\theta'$.
 \item If $u=f(u_1,\dots,u_k)\wedge v=g(v_1,\dots,v_m)$ then $f=g$, $m=k$ and for all $i\leq k$ we have $v_i\theta=u_i\theta$.
  It is enough to prove that $v_i\theta'=u_i\theta'$.
  Let us prove this case by induction on $\min(\termheight{u},\termheight{v})$. For the basis of induction,  we have that either $u_i\in\vs$ or $v_i\in\vs$ 
  (otherwise the basis is not minimal) and we have proved already that $v_i\theta=u_i\theta\implies v_i\theta'=u_i\theta'$.
  Suppose the statement is true for $\min(\termheight{u},\termheight{v})\leq n$. For $\min(\termheight{u},\termheight{v})= n+1$
   we have $\termheight{u_i\theta}\leq n$, $\termheight{v_i\theta}\leq n$ and $u_i\theta=v_i\theta$  for all $i$. 
   Then by induction supposition and two cases considered before we have $u_i\theta'=v_i\theta'$.
\end{itemize}

   Thus, $\forall u,v\in\Sub{U} u\theta=v\theta\implies u\theta'=v\theta'$,
   i.e. $\theta'$ is a unifier of $U$. Moreover, for all $x\in\dom{\theta}, x\theta=(x\theta')\gamma$, where $\gamma=\set{z\mapsto p}$.
   
   Since $\theta$ is a most general unifier, we have $p\in\vs$ which contradicts to $p\in\Sub{x\theta | x\in\dom{\theta}}\setminus \vs$.

 
\end{proof}
\end{lemma}

\begin{lemma}\label{lemma:localization:one-to-one}
  Let $\css$ be a constraint system, $\sigma$ be a ground
  substitution such that $\sigma \models \css$.

  Then there exists
  a set of terms $T$, 
 a substitution $\tau$ of domain $\Var{T}\setminus \Var{\css}$,
 a substitution $\theta$ and 
 a $(\css\theta,\sigma)$-compliant derivation $D$ 
 such that 
 \begin{itemize*}
 \item $D$ is $(T,\sigma\cup\tau)$-maximal and one-to-one localized by $T$ for  $\sigma\cup\tau$
 \item $\sigma\cup\tau = \theta(\sigma\cup\tau)$
 \item $\Sub{\css\theta}\subseteq T$
 \item $T$ and $\theta$ of size linear in $\card{\Sub{\css}}$ 
\end{itemize*}

\begin{proof}
  Under the same assumptions, by Lemma~\ref{lemma:localization}, there
  exists $T_0\supseteq \scss$ of size linear in \card{\Sub{\css}} and $\tau$ of
  domain $\Var{T_0}\setminus\Var{\css}$ such that there exists a
  $(T_0,\sigma\cup\tau)$-maximal and $(\css,\sigma)$-compliant
  derivation $D$ which is localized by $T_0$ for the same substitution
  $\sigma'=\sigma\cup\tau$. 

  Let $\mathcal{U}=\condset{ t\unif t'}{t,t'\in \Sub{T_0}
    \text{ and }t\sigma'=t'\sigma'}$. The unification system
  $\mathcal{U}$ has a unifier $\sigma'$ and thus has a most
  general solution $\theta$. 
  By \Cref{:mgu|is|a|unit}, $\sigma' = \theta\sigma'$.

  Let $T=\Sub{T_0}\theta$.

  Since $\Sub{\css}\subseteq T_0$ we have $\Sub{\css\theta}\subseteq
  \Sub{T_0\theta}$.
  Since $\theta$ is a most general unifier of $\mathcal{U}$
  and $\Sub{\mathcal{U}}= \Sub{T_0}$ we have 
  $\Sub{T_0\theta} = \Sub{T_0}\theta$ by \Cref{:mgu|sub}.
  This implies 
  \textit{(i)} $\Sub{\css\theta}\subseteq T$, 
  \textit{(ii)} $\theta$ is of linear size on $\card{\Sub{T_0}}$ and thus on $\card{\Sub{\css}}$,
  and \textit{(iii)} $T$ is of linear size on $\card{\Sub{\css}}$.
  Moreover, as  $\sigma' = \theta\sigma'$ we have $\Sub{T}\sigma' = \Sub{T_0}\sigma'$ and thus
  from $D$ is $(T_0,\sigma')$-maximal follows $D$ is $(T,\sigma')$-maximal.

  Assume there exists $t,t' \in \Sub{T}$ such that
  $t\sigma'=t'\sigma'$ but $t\neq t'$. 
  Since $T=\Sub{T_0\theta}$ there exists $t_0,t_0'\in \Sub{T_0}$ such
  that $t_0\theta \neq t_0'\theta$ but $t_0\theta\sigma' =
  t_0'\theta\sigma'$.
  From $\sigma' = \theta\sigma'$
  we have an
  existence of  $t_0,t_0'\in \Sub{T_0}$ such
  that $t_0\theta \neq t_0'\theta$ but $t_0\sigma' =
  t_0'\sigma'$.
  This contradicts the fact that  $\theta$ 
  satisfies $\mathcal{U}$.  
  
  Finally, from $D$ is $(\css,\sigma)$-compliant and $\sigma = \theta\sigma$ we have $D$ is $(\css\theta, \sigma)$-compliant.
\end{proof}
\end{lemma}

\subsection{Milestone sequence}
\label{subsubsec:derivation:milestones}

In addition to retrace the deduction steps performed in $D$ we want to
track which terms relevant to $\css$ are deduced in $T$, and in which
order.

\begin{definition}[Milestone sequence\label{def:milestone}]
  Let $T$ be a set of terms and $\sigma$ be a ground substitution.  We
  say that $\vec T$ is the $(T,\sigma)$-milestone sequence of a
  derivation $D=(l_i \to r_i)_{1\le i \le m}$ if $\vec T = t_1,\ldots
  , t_n$ is a sequence of maximal length in which each $t_i$ is either
  of the form $\to t$ or of the form $\hsend{t}$, with $t\in\Sub{T}$
  and there exists a strictly increasing function $\alpha:
  \set{1,\ldots,n}\to\set{1,\ldots,m}$ such that for every $1\le i\le
  n$ we have:
  \begin{enumerate}
  \item if $\vec T[i]=\hsend{t}$ then $D[\alpha(i)] = \hsend{t\sigma}$; 
  \item if $\vec T[i] = \to t$ then $D[\alpha(i)] = l_i\to t\sigma$ is a standard
    deduction rule; 
  \end{enumerate}
\end{definition}

\begin{lemma}\label{:all|vars|in|milestone}
 Let $\sigma\models \css$, $T\supseteq \scss$ and $\sigma'$ be an extension of $\sigma$ on $\vars{T}$.
 Let $D$ be $(T,\sigma')$-maximal derivation one-to-one localized by $T$ for $\sigma'$.
 Let $\vec T$ be a $(T,\sigma')$-milestone sequence. 
 \textbf{Then} for any $i$ for any $x\in\vars{\vec T[i]}$ there exists $j<i$ such that $\vec T[j] = \to x$. 
\end{lemma}
 
\begin{proof}
	If $x\in\vars{\vec T[i]}$ then there exists corresponding deduction  $D[j]$ that deduces term $\vec T[i]\sigma'$.  
	Then by  \Cref{lemma:variable|created|by|composition} there exists $k<j$ such that $D[j]$ deduces by  a standard rule $x\sigma'$.
	From the injectivity of $\sigma$ follows that $x$ is the only term of $\Sub{T}$ having $\sigma'$ image equal $x\sigma'$.
	Thus, by definition of milestone sequence, there exists $m<i$ such that $\vec T[m] = \to x$.
\end{proof}

\section{Deciding constraint systems}
\label{sec:deciding}
From now we suppose that the \textbf{considered subterm deduction system} contains
 a rule $x_1, x_{2} \rightarrow f(x_1, x_{2})$, where $f$ is a function symbol with arity 
2  that does not occur in any other rule.


\begin{theorem}\label{:bounded|solution}
 Let $\sigma$ such that $\sigma\models \css$, $T$ such that $T\supseteq \scss$ 
and $\sigma'$ an extension of $\sigma$ on $\vars{T}$.
 Let $D$ be a $(T,\sigma')$-maximal derivation one-to-one localized by $T$ for $\sigma'$ in which no term is deduced twice by a standard rule.
 
 Then there exists a solution $\tau$ of $\css$ of size polynomial in \card{\Sub{T}}.
\begin{proof} 
First let us define a \emph{replacement} of a term $q$ by term $p$ in $t$
denoted as $\repl{t}{q \mapsfrom p}$ as follows: $\repl{t}{q \mapsfrom
  p}$ is the term is obtained from $t$ by simultaneous replacing all
occurrences of $q$ in $t$ by $p$.  For a substitution
$\sigma=\condset{x\mapsto t_x}{x\in\dom{\sigma}}$ we define
$\repl{\sigma}{q \mapsfrom p} = \condset{x\mapsto (\repl{t_x}{q
    \mapsfrom p})}{x\in\dom{\sigma}}$

 Let $\vec T$ be a $(T,\sigma')$-milestone sequence for $D$. 

 Let $\vec M=m_1,\dots, m_n$ be the maximal increasing sequence such that for any $i=1,\dots,n$,  $\vec T[m_i] = \hsend{t_{m_i}}$.
 We put also $m_0=0$ and $m_{n+1} = \card{\vec T } + 1$.
 Let $\vec T_i = \vec T[m_i + 1 : m_{i+1}-1]$.

 \smallpar{Goal.}
 We will prove the existence of a ground substitution $\tau'$, set of terms $T'\supseteq T$
 and a derivation $D'$ which is $(\css,\tau)$-compliant, $(T',\tau)$-maximal (where $\tau=\tau'|_{\vars{\css}}$  is of a linear size on $\Sub{T}$)
 and is one-to-one localized by $T'$ with $\tau'$
 such that its $(T,\tau')$-milestone sequence coincides with $\vec T$.

 If it is proved, by \Cref{:model|criteria} we can show that $\tau \models \css$.

 \smallpar{Build $T'$.} 
 Let  $X$ be the set of variables of $\vec T$  whose $\sigma'$-instance are not derivable from the empty knowledge.
 By \Cref{:all|vars|in|milestone} 
 each variable $x$ of $\vars{\vec T}$ appears first as $\to x$ in $\vec T$.   
 Therefore,  we may put  
 $X=\set{x_1,\dots,x_u} =  \vars{\vec T} \setminus \set{x: \to x \in  \vec T_0}$.
 Let for each $x\in X$, let $\bar{x}$ be a new fresh variable (corresponding to $x$) and let $\bar X = \set{\bar x : x\in X}$. 
 Finally, we put $T'=T\cup \bar X$.

 \smallpar{Build $\tau'$.} 
 Let $\tau'$ be a ground substitution defined as follows: 
 \begin{itemize*}
  \item 
 for any $x \in X$, $\bar x\tau'$ is a nonce $n_{x}$ and  $x\tau' = f(t_{m_i}\tau',n_x)$, where $\to x$ appears first in $\vec T_{i}$
 (note that by \Cref{:all|vars|in|milestone} for any $y \in\vars{t_{m_i}}$, $\to y$ appears first time at position before $m_i$ in $\vec T$
 and thus $\tau'$ is correctly defined);
  \item
 for any $y \in \vars{\vec T_0}$, $y\tau' = n_y$; 
  \item
 \footnote{We note that in practice $\vars{T}\setminus \vars{\vec T} = \emptyset$ if we see how $T$ is constructed in \Cref{lemma:localization:one-to-one}.}for any $z\in\vars{T}\setminus \vars{\vec T}$, $z\tau' = a_z$, where $a_z$ is a fresh  constant from $\as \setminus \Nonces$ not appearing in \Sub{T}. 
 \end{itemize*}

  We can see that $x\tau$ is of polynomial size on $\card{\Sub{\css}}$ for any $x\in\vars{T}$. 
  \smallpar{Show $\tau'$ is injective on $\Sub{T'}$}.
    Suppose the contrary, let $p,q\in \Sub{T'}$ be a pair with minimal
    size of $p\tau'$ and having $p\tau'=q\tau'$, while $p\neq q$.  If
    neither $p$ nor $q$ is a variable, then this contradicts the
    minimality of $p\tau'$ (we can choose subterms of $p$ and $q$
    satisfying the choice criteria).  If both are variables, then it
    is not possible by the construction of $\tau'$.  W.l.o.g. let
    $p\in \vs$ and $q\notin \vs$. The case where $p\tau'$ is a nonce
    or another constant is impossible; thus $p\tau' = f(t_{i_j},n_x)$
    and $q = f(u, \bar{x})$ (since by construction for every nonce
    $n_x$ there exists only one variable $\bar{x}$ such that
    $\bar{x}\tau'=n_x$ and $n_x\notin \Sub{T'}$).  But again, by
    construction (note that $\bar{x}$ was a fresh variable), the only
    term in $\Sub{T'}$ having $\bar{x}$ as a subterm is $\bar{x}$,
    thus $q\in\vs$: contradiction.\hfill$\Diamond$
 
  \smallpar{Build a replacement to pass from $\tau'$ to $\sigma'$.}
  Let $\delta$ be the replacement $\delta = |_{\condset{x\tau'
      \mapsfrom x\sigma'}{x\in\vars{T}}}$.  Then $\tau'\delta =
  \sigma'$ on $\vars{T}$. Moreover, from the property we have just
  proven follows that for any $t\in\Sub{T'}$, we have $(t\tau')\delta
  =t (\tau' \delta)$ and for $t\in\Sub{T}$, we have $(t\tau')\delta =t
  \sigma'$.  Note also that $(x\tau')\delta = x\tau'$ for any
  $x\in\bar X$.

 \smallpar{Build  $(\css,\tau)$-compliant derivation $D'$ localized by $T'$ with $\tau'$.} 
 Let $D'_0 = \hnoncereceive {n_{x_1}}, \dots, \hnoncereceive {n_{x_u}}$.
 Let $D'_1$ be a sequence of rules of length $\card{\vec T}$ such that for any $i \leq \card{\vec T}$:
 \begin{itemize*}
	 \item if $\vec T[i] = \hsend{t}$ then $D'_1[i] = \hsend{t\tau'}$;
	 \item if $\vec T[i] = \to x$ and $x\in X$ then $D'_1[i] = n_x, t_{m_j}\tau' \to x\tau'$, where $\to x$ appears first in $\vec T_{j}$; 
	 \item if $\vec T[i] = \to y$ and $y\in \vars{\vec T_{0}}$ then $D'_1[i] = \hnoncereceive {y\tau'}$; 
	 \item if $\vec T[i] = \to {t}$ and $t\notin \vs$ then 
		 since $D$ is one-to-one localized by $T$, there exists $t_1,\dots, t_k$ such that 
		 $\hsend{t_j} \in \vec T[1:i-1]$ or ${\to} t_j \in \vec T[1:i-1]$ for $j=1,\dots,k$ and 
		 $t_1,\dots, t_k \to t$ is a deduction rule. Thus, we put $D'_1[i] = t_1\tau',\dots, t_k\tau' \to t\tau'$.
 \end{itemize*}

 We define $D'= D'_0, D'_1$. Note that $\RHSset{D'_0}{} = \bar X\tau'$
 and for any $i$, $\RHSset{D'_1}{i} = \vec T[1:i]\tau'$. Thus, by the
 construction $D'$ is a derivation which is $(\css,\tau)$-compliant
 and localized by $T'$ for $\tau'$. Moreover, it is one-to-one
 localized since $\tau'$ is injective on $\Sub{T'}$.


 We have by construction of $D'$ that its $(T', \tau')$-milestone sequence  is $\vec T' = \to \bar x_1,\dots,\to\bar x_u,\vec T$.
Moreover, $\card{D'}=\card{\vec T'}$.


 \smallpar{Show that $D'$ is $(T',\tau')$-maximal.} That is,
 for any $t\in\Sub{T'}$ if 
$t\tau'\in\Der{\RHSset{D'}{i}}$ then 
$t\tau'\in\RHSset{D'}{\Next{D'}{i}-1}$.
 
The case $t\in\bar X$ is trivial, since $\bar X\tau'$ is deduced at the very beginning of $D'$.
 
\begin{removable}
Note also that by construction and since $D$ is maximal we have that  \cmnt{do we need it?}
  for any $j$ any term 
 $t\in T$ deducible from ${t_{i_1}\sigma', \dots, t_{i_j}\sigma'}$
 is also deducible from $t\tau' \in \Der{t_{i_1}\tau', \dots, t_{i_j}\tau'}$
 and deduced in $D'$ before $(j+1)$-th send index.  
\end{removable}

 Suppose that there an exists index $j$ and term $t\in \Sub{T}$ such that $t\tau' \in \Der{t_{m_1}\tau', \dots, t_{m_j}\tau'}$ 
 but $t\tau' \notin \RHSset{D'}{u+m_{j+1}-1}$, i.e. $t\tau'$ is not deduced before the next to $j$ non-standard rule in $D'$.
 In this case,  $t\sigma'\notin\Der{t_{m_1}\sigma', \dots, t_{m_j}\sigma'}$, otherwise by maximality $t\sigma'$
 would be deduced before $(j+1)$-th non-standard rule of $D$ and by construction, $t\tau'$ would also appear 
 in $D'$ before $(j+1)$-th nonstandard rule of $D'$.

  Let $j$ be such a minimal index. 
  Note that $\vars{t} \subseteq \vars{\vec T}$, otherwise by construction $t\tau'$ would contain some fresh constants from $\as\setminus\Nonces$ and thus would not be derivable from $t_{i_1}\tau', \dots, t_{i_j}\tau'$. 
  Let $m'$ (resp. $m$) be the maximal index such that $D'[1:m']$ (resp. $D[1:m]$) contains exactly $j$ non-standard rules.
  Thus, $t\tau'\in \Der{\RHSset{D'}{m'}}$ and $t\sigma' \notin \Der{\RHSset{D}{m}}$. 
  Note that $t\tau'\notin \RHSset{D'}{m'}$ (otherwise it would imply $t\sigma'\in \RHSset{D}{m}$). 
  Let $E'$ be a minimal sequence of standard rules such that
  $D'[1:m'],E'$ is a derivation ending with a standard deduction of $t\tau'$.
  W.l.o.g.,  we suppose that $E'[1:\card{E'}-1]$ does not deduce terms from $\Sub{T}\tau'$ 
  (otherwise, if $t'\tau'$ is deduced in $E'[1:\card{E'}-1]$ with $t'\in\Sub{T}$ then \textit{(i)} either $t'\sigma'\in\Der{\RHSset{D}{m}}$ and by maximality of $D$ 
   $t'\sigma'\in{\RHSset{D}{m}}$ which contradicts the minimality of $E'$ \textit{(ii)} or $t'\sigma'\notin\Der{\RHSset{D}{m}}$ which implies
   $t'\sigma'\notin{\RHSset{D}{m}}$; thus by construction $t'\tau' \notin{\RHSset{D'}{m'}}$ and we could chose $t'$ instead of $t$).

  Let $\css'$ be a constraint system obtained from $\css$ by removing all constraints after $j$-th $\snd{}{}$-constraint and removing all $\forbid{}{}$-constraints.
  By construction, $D'[1:m'], E'$ is a proof of $\tau \models \css'$ and thus we can apply \Cref{lemma:variable|created|by|composition},
  i.e. all rules of $E'[1:\card{E'}-1]$  are compositions.

  Suppose that $t$ is a variable. Note that $t\tau'$  is not a nonce, otherwise by definition of $\tau'$, $t\in \vec T_0$ and thus $t\sigma'\in\RHSset{D}{m}$. 
  Therefore, $t\tau'=f(t_{m_k}\tau', n_t)$,   where $\to t$ first appears in $\vec T_k$.
  Since $t$ is a variable, the last rule of $E'$ is also a composition, more precisely $t_{m_k}\tau', n_t \to f(t_{m_k}\tau', n_t)$. 
  If $k\leq j$, by construction of $\tau'$, $t\sigma'$ must be in $\RHSset{D}{m}$.
  Thus,  $k> j$.
  Since $D'[1:m'], E'$ is a derivation, either $t_{m_k}\tau' \in \RHSset{E'}{\card{E'}-1}$
  or $t_{m_k}\tau' \in \RHSset{D'}{m'}$. The former contradicts the choice of $E'$. 
  The latter case implies $t_{m_k}\sigma' \in \RHSset{D}{m} \subseteq \Der{t_{m_1}\sigma', \dots, t_{m_j}\sigma'}$ and thus, 
  as $j<k$ we have that $\Der{t_{m_1}\sigma', \dots, t_{m_k}\sigma'}= \Der{t_{m_1}\sigma', \dots, t_{m_{k-1}}\sigma'}$.
  Thus, $\vec T_k$ must be empty, otherwise it contradicts the maximality of $D$ and that no term is deduced twice by a standard rule in $D$. 
  This contradicts that $\to t$ appears first in $T_k$.

  Thus, $t\notin \vs$.

  Let us build 
  a sequence of rules $E$ such that $E[i]=E'[i]\delta$  
  and show that $D[1:m],D'_0,E$ is a proof of
  $t\sigma'\in\Der{t_{m_1}\sigma', \dots, t_{m_j}\sigma'}$.  
  
  Let us show that $E'[i]\delta$ is a rule.
  \begin{itemize*}
   \item If $E'[i]=\hnoncereceive{o}$ is a nonce generation, then $o\notin\img{\tau'}$ 
     due to the minimality of $E'$ and since all variables of $T'$ that are mapped to nonces by $\tau'$ are deduced in $D'$ before the first  
     non-standard rule. Thus $o\delta = o$ and we have $E[i]=\hnoncereceive{o}$
   \item If $E'[i]$ is another composition, then $E'[i] = t'_1,\dots, t'_v \to h(t'_1,\dots,t'_v)$.
         Since $t\notin \vs$ 
         and $E'[1:\card{E'}-1]$ does not deduce terms from $\Sub{T}\tau'$ 
         we have $h(t'_1,\dots,t'_v) \neq x\tau'$
         for any $x\in\vars{T'}$. Thus, $h(t'_1,\dots,t'_v)\delta = h(t'_1\delta,\dots,t'_v\delta)$ and we have 
         $t'_1\delta,\dots,t'_v\delta \to h(t'_1\delta,\dots,t'_v\delta)$ is a composition rule.
   \item If $E'[i]$ is a decomposition, then 
	 since no decomposition rule contains $f$, the value of $x\tau'$ (which is a fresh nonce or has $f$ as a root symbol)
	may be replaced with any other term and we still obtain an instance of the same decomposition rule,
	i.e. $E'[i]\delta$ is an instance of a decomposition rule.
       
  \end{itemize*}
  
  As noted above, since $\forall r\in\Sub{T}, (r\tau')\delta = r\sigma'$ and $D'_0\delta=  D'_0$
  by construction we have 
  $\RHSset{D'}{m'}\delta \subseteq   \RHSset{D}{m} \cup \set{x\tau': x\in \bar X}$. 
  Thus, $D[1:m],D'_0, E$ is a derivation deducing  $t\tau'\delta = t\sigma'$, i.e.
  $t\sigma'\in\Der{t_{i_1}\sigma', \dots, t_{i_j}\sigma'}$. Contradiction.

  Therefore, $D'$ is $(T',\tau')$-maximal.


\smallpar{Conclusion.}
Since $\Sub{\css}\subseteq T'$, and $\tau'$ is injective on $\Sub{T'}$,
we have that 
by construction of $D'$, for any term $t\in\Sub{\css}$, $t\sigma'$ is deduced before $j$-th non-standard rule of $D$ (resp. deduced in $D$) if and only if 
  $t\tau'$ is deduced before $j$-th non-standard rule of $D'$ (resp. deduced in $D'$).
Therefore, 
since $\sigma\models\css$ and $D$ is $(\css,\sigma)$-compliant and $(\scss,\sigma)$-maximal 
and 
since $D'$ is $(\css,\tau)$-compliant and $(\scss,\tau)$-maximal 
we may use 
twice \Cref{:model|criteria} and obtain that  $\tau$ satisfies $\css$. 
\end{proof}

\end{theorem}

\begin{corollary}\label{:iff}
 Let $\css$ be a constraint system.
 $\css$ is satisfiable, if and only if there exists a solution $\sigma'$ of $\css$ with 
polynomial size w.r.t.  $\card{\Sub{\css}}$.
\begin{proof}
 $(\Leftarrow)$ is trivial, since $\sigma' \models \css$. Consider
 $(\Rightarrow)$. Let $\sigma\models \css$.
 By \Cref{lemma:localization:one-to-one} there exists a set of terms $T$,
 a substitution $\theta$ both with the size linear in $\card{\scss}$
 and an
 extension $\gamma$ of $\sigma$ and $(T,\gamma)$-maximal $(\css\theta,\sigma)$-compliant derivation $D$ one-to-one localized by $T$ for $\gamma$.
 We also have  $\gamma = \theta \gamma$ (which implies $\sigma=\theta\sigma$). Thus $\sigma$ satisfies  $\css\theta$.
 
 From the same lemma we have $\Sub{\css\theta}\subseteq T$. 
 By \Cref{:bounded|solution} there exists a substitution $\tau$ of size polynomial in $\card{\Sub{T}}$ (and consequently, polynomial in $\card{\Sub{\css}}$)  such that $\tau\models \css\theta$.
 From this we have $\theta\tau \models \css$. Moreover, since both $\theta$ and $\tau$ are of polynomial size on $\card{\Sub{\css}}$,
 $\sigma'=\theta\tau$ is also of polynomial size on $\card{\Sub{\css}}$ and $\sigma' \models \css$.
\end{proof}
\end{corollary}

From the previous result we can directly derive an NP  decision 
procedure for constraint systems satisfiability:  guess a substitution 
of polynomial size in $\card{\Sub{\css}}$ and check whether it satisfies $\css$ 
in polynomial time  (see e.g. \cite{AbadiCortier}). 



\section{Conclusion}
\label{sec:conclusion}
We have obtained the first decision procedure for deducibility constraints with negation  and
we have applied it to the synthesis of mediators subject to non-disclosure
policies. It has been implemented as an extension of
CL-AtSe~\cite{Turuani06} for the Dolev-Yao deduction system. On the
Loan Origination case study, the prototype generates directly the
expected orchestration. Without negative constraints undesired
solutions in which the mediator impersonates the clerks were found.
More details, including problem specifications, can be found at \url{http://cassis.loria.fr/Cl-Atse}.
As in~\cite{AbadiCortier,Baudet} our definition of subterm deduction systems can be
extended to allow ground terms in right-hand sides of decomposition
rules even when they are not subterms of left-hand sides and 
the decidability result
remains valid with minor adaptation of the proof.  A more challenging
extension would be to consider general constraints (as in~\cite{ACIarxiv}) with negation.

\bibliographystyle{plain}
\bibliography{neg,aci,avantssar,4these,rfc}

\ifdraft\end{document}\fi
\end{document}

\appendix

\section{The ASLan++ specification}

\begin{lstlisting}[breaklines]
specification NegativeConstraints
channel_model CCM

entity Environment {

  symbols 
	nonpublic noninvertible  g(agent)    : symmetric_key ;
	noninvertible h(message)  : message ;
	loan , request,amount100,amount1000: text ;	
	client , alice, bob ,charlie,pep,mediator : agent ;

  entity Client(Actor: agent,
	P, M: agent ,
	Amount: text  ) {
    symbols
	 Ephemeral_key : symmetric_key ;
	 Resp_A,Resp_B,
	 End_execution_client: text ;
         A,B : agent ;

    body {
      Actor -> M : {g(Actor).loan.P}_inv(pk(Actor)) ;
      M -> Actor : ?A.?B ;
                  Ephemeral_key := fresh() ;
      Actor -> M : {Amount.Actor.Ephemeral_key}_pk(A).{Amount.Actor.Ephemeral_key}_pk(B) ;
      M -> Actor : {h(A.Amount.Actor.?Resp_A)}_inv(pk(A)).{|?Resp_A|}_Ephemeral_key. 
		   {h(B.Amount.Actor.?Resp_B)}_inv(pk(B)).{|?Resp_B|}_Ephemeral_key ;
      secrecy_End_of_execution:(End_execution_client) := fresh() ;
      Actor -> P: {Amount.Actor.A.Resp_A.B.Resp_B}_pk(P). 
		   {h(A.Amount.Actor.Resp_A)}_inv(pk(A)). 
	           {h(B.Amount.Actor.Resp_B)}_inv(pk(B)).End_execution_client ;
    }

  goals
    secrecy_End_of_execution:(_) {Actor,M};

  }

  entity Clerk(Actor:agent) {
    symbols
	M,Client: agent ;
	Ephemeral_key : symmetric_key;
	Amount,Response : text ;

    body {
      ? -> Actor : request.?M ;
      Actor -> M : g(Actor).pk(Actor) ;
      M -> Actor : {?Amount.?Client.?Ephemeral_key}_pk(Actor) ;
                  Response := fresh() ;
      Actor -> M : {h(Actor.Amount.Client.Response)}_inv(pk(Actor)).{|Response|}_Ephemeral_key ;
    }
  }

  entity Mediator(Actor:agent,M:agent) {
    body {
      Actor -> M: {|g(charlie)|}_g(bob). {|g(bob)|}_g(charlie).{|g(client)|}_g(bob).{|g(bob)|}_g(client) ;
    }
  }

  body {
     new Clerk(alice) ;
     new Clerk(bob) ;
     new Clerk(charlie) ;
     new Client(client,pep,mediator,amount100) ;
     new Mediator(i,alice) ;
  }

  goals
    secrecy_End_of_execution:(_) {};
}
\end{lstlisting}

\section{The ASLan specification}

\begin{lstlisting}[breaklines]
% @specification(NegativeConstraints)
% @channel_model(CCM)
% @connector_name(ASLan++ Connector)
% @connector_version(0.6.8)
% @connector_options(-opt LUMP -hc ALL -gas)

section signature:

	text > slabel
	ak : agent -> public_key
	child : nat * nat -> fact
	ck : agent -> public_key
	defaultPseudonym : agent * nat -> public_key
	descendant : nat * nat -> fact
	dishonest : agent -> fact
	g : agent -> symmetric_key
	h : message -> message
	hash : message -> message
	pk : agent -> public_key
	secrecy_End_of_execution_set : nat -> set(agent)
	secret_End_of_execution_set : nat -> set(agent)
	state_Clerk : agent * nat * nat * agent * agent * symmetric_key * text * text -> fact
	state_Client : agent * nat * nat * agent * agent * text * symmetric_key * text * text * text * agent * agent -> fact
	state_Environment : agent * nat * nat -> fact
	state_Mediator : agent * nat * nat * agent -> fact
	succ : nat -> nat

section types:

	A : agent
	% @original_name(name=A; match=true)
	A_1 : agent
	Actor : agent
	Ak_arg_1 : agent
	Amount : text
	B : agent
	% @original_name(name=B; match=true)
	B_1 : agent
	Ck_arg_1 : agent
	Client : agent
	% @original_name(name=Client; match=true)
	Client_1 : agent
	Descendant_arg_1 : nat
	Descendant_arg_2 : nat
	Descendant_arg_3 : nat
	Dummy : agent
	% @original_name(name=Actor)
	E_C_Actor : agent
	% @original_name(name=Actor)
	E_C_Actor_ : agent
	% @original_name(name=Amount)
	E_C_Amount : text
	% @original_name(name=Amount; match=true)
	E_C_Amount_1 : text
	% @original_name(name=Ephemeral_key)
	E_C_Ephemeral_key : symmetric_key
	% @original_name(name=Ephemeral_key; match=true)
	E_C_Ephemeral_key_1 : symmetric_key
	% @original_name(name=IID)
	E_C_IID : nat
	% @original_name(name=IID)
	E_C_IID_ : nat
	% @original_name(name=M)
	E_C_M : agent
	% @original_name(name=M; match=true)
	E_C_M_1 : agent
	% @original_name(name=SL)
	E_C_SL : nat
	% @original_name(name=SL)
	E_C_SL_ : nat
	% @original_name(name=Actor)
	E_M_Actor : agent
	% @original_name(name=IID)
	E_M_IID : nat
	% @original_name(name=M)
	E_M_M : agent
	% @original_name(name=SL)
	E_M_SL : nat
	% @original_name(name=Knowers)
	E_sEoe_Knowers : set(agent)
	% @original_name(name=Msg)
	E_sEoe_Msg : message
	End_execution_client : text
	% @original_name(name=End_execution_client; fresh=true)
	End_execution_client_1 : text
	Ephemeral_key : symmetric_key
	% @original_name(name=Ephemeral_key; fresh=true)
	Ephemeral_key_1 : symmetric_key
	H_arg_1 : message
	Hash_arg_1 : message
	IID : nat
	% @original_name(name=IID; fresh=true)
	IID_1 : nat
	% @original_name(name=IID; fresh=true)
	IID_2 : nat
	% @original_name(name=IID; fresh=true)
	IID_3 : nat
	% @original_name(name=IID; fresh=true)
	IID_4 : nat
	% @original_name(name=IID; fresh=true)
	IID_5 : nat
	Knowers : set(agent)
	M : agent
	Msg : message
	P : agent
	Pk_arg_1 : agent
	Resp_A : text
	% @original_name(name=Resp_A; match=true)
	Resp_A_1 : text
	% @original_name(name=Resp_A; match=true)
	Resp_A_2 : text
	Resp_B : text
	% @original_name(name=Resp_B; match=true)
	Resp_B_1 : text
	% @original_name(name=Resp_B; match=true)
	Resp_B_2 : text
	Response : text
	% @original_name(name=Response; fresh=true)
	Response_1 : text
	SL : nat
	Succ_arg_1 : nat
	alice : agent
	amount100 : text
	amount1000 : text
	atag : slabel
	bob : agent
	charlie : agent
	client : agent
	ctag : slabel
	dummy_agent_1 : agent
	dummy_agent_2 : agent
	dummy_agent_3 : agent
	dummy_agent_4 : agent
	dummy_agent_5 : agent
	dummy_agent_6 : agent
	dummy_agent_7 : agent
	dummy_agent_8 : agent
	dummy_nat : nat
	dummy_symmetric_key_1 : symmetric_key
	dummy_symmetric_key_2 : symmetric_key
	dummy_symmetric_key_3 : symmetric_key
	dummy_symmetric_key_4 : symmetric_key
	dummy_text_1 : text
	dummy_text_2 : text
	dummy_text_3 : text
	dummy_text_4 : text
	dummy_text_5 : text
	dummy_text_6 : text
	dummy_text_7 : text
	dummy_text_8 : text
	dummy_text_9 : text
	% @original_name(name=request)
	e_request : text
	false : fact
	loan : text
	mediator : agent
	pep : agent
	root : agent
	secrecy_End_of_execution : protocol_id
	secret_End_of_execution : protocol_id
	stag : slabel
	true : fact

section inits:

% @new_instance(new_entity=Environment; Actor=root; IID=0; SL=1)
initial_state init :=
	child(dummy_nat, 0).
	dishonest(i).
	iknows(0).
	iknows(alice).
	iknows(amount100).
	iknows(amount1000).
	iknows(atag).
	iknows(bob).
	iknows(charlie).
	iknows(client).
	iknows(ctag).
	iknows(e_request).
	iknows(i).
	iknows(inv(ak(i))).
	iknows(inv(ck(i))).
	iknows(inv(pk(i))).
	iknows(loan).
	iknows(mediator).
	iknows(pep).
	iknows(root).
	iknows(stag).
	iknows(pair(scrypt(g(bob), g(charlie)), pair(scrypt(g(charlie), g(bob)), pair(scrypt(g(bob), g(client)), scrypt(g(client), g(bob)))))).
	state_Environment(root, 0, 1).
	true

section hornClauses:

hc public_ck(Ck_arg_1) :=
	iknows(ck(Ck_arg_1)) :-
		iknows(Ck_arg_1)

hc public_ak(Ak_arg_1) :=
	iknows(ak(Ak_arg_1)) :-
		iknows(Ak_arg_1)

hc public_pk(Pk_arg_1) :=
	iknows(pk(Pk_arg_1)) :-
		iknows(Pk_arg_1)

hc public_hash(Hash_arg_1) :=
	iknows(hash(Hash_arg_1)) :-
		iknows(Hash_arg_1)

hc public_succ(Succ_arg_1) :=
	iknows(succ(Succ_arg_1)) :-
		iknows(Succ_arg_1)

hc inv_succ_1(Succ_arg_1) :=
	iknows(Succ_arg_1) :-
		iknows(succ(Succ_arg_1))

hc descendant_closure(Descendant_arg_1, Descendant_arg_2, Descendant_arg_3) :=
	descendant(Descendant_arg_1, Descendant_arg_3) :-
		descendant(Descendant_arg_1, Descendant_arg_2),
		descendant(Descendant_arg_2, Descendant_arg_3)

hc descendant_direct(Descendant_arg_1, Descendant_arg_2) :=
	descendant(Descendant_arg_1, Descendant_arg_2) :-
		child(Descendant_arg_1, Descendant_arg_2)

hc public_h(H_arg_1) :=
	iknows(h(H_arg_1)) :-
		iknows(H_arg_1)

section rules:

% @new_instance(entity=Environment; iid=IID; line=61; new_entity=Clerk; Actor=alice; IID=IID_1; SL=1; M=dummy_agent_1; Client=dummy_agent_2; Ephemeral_key=dummy_symmetric_key_1; Amount=dummy_text_1; Response=dummy_text_2)
% @new_instance(entity=Environment; iid=IID; line=62; new_entity=Clerk; Actor=bob; IID=IID_2; SL=1; M=dummy_agent_3; Client=dummy_agent_4; Ephemeral_key=dummy_symmetric_key_2; Amount=dummy_text_3; Response=dummy_text_4)
% @new_instance(entity=Environment; iid=IID; line=63; new_entity=Clerk; Actor=charlie; IID=IID_3; SL=1; M=dummy_agent_5; Client=dummy_agent_6; Ephemeral_key=dummy_symmetric_key_3; Amount=dummy_text_5; Response=dummy_text_6)
% @new_instance(entity=Environment; iid=IID; line=64; new_entity=Client; Actor=client; IID=IID_4; SL=1; P=pep; M=mediator; Amount=amount100; Ephemeral_key=dummy_symmetric_key_4; Resp_A=dummy_text_7; Resp_B=dummy_text_8; End_execution_client=dummy_text_9; A=dummy_agent_7; B=dummy_agent_8)
% @new_instance(entity=Environment; iid=IID; line=65; new_entity=Mediator; Actor=i; IID=IID_5; SL=1; M=alice)
% @step_label(entity=Environment; iid=IID; line=65; variable=SL; term=6)
step step_001_Environment__line_61(Actor, IID, IID_1, IID_2, IID_3, IID_4, IID_5) :=
	state_Environment(Actor, IID, 1)
	=[exists IID_1, IID_2, IID_3, IID_4, IID_5]=>
	child(IID, IID_1).
	child(IID, IID_2).
	child(IID, IID_3).
	child(IID, IID_4).
	child(IID, IID_5).
	state_Clerk(alice, IID_1, 1, dummy_agent_1, dummy_agent_2, dummy_symmetric_key_1, dummy_text_1, dummy_text_2).
	state_Clerk(bob, IID_2, 1, dummy_agent_3, dummy_agent_4, dummy_symmetric_key_2, dummy_text_3, dummy_text_4).
	state_Clerk(charlie, IID_3, 1, dummy_agent_5, dummy_agent_6, dummy_symmetric_key_3, dummy_text_5, dummy_text_6).
	state_Client(client, IID_4, 1, pep, mediator, amount100, dummy_symmetric_key_4, dummy_text_7, dummy_text_8, dummy_text_9, dummy_agent_7, dummy_agent_8).
	state_Environment(Actor, IID, 6)

% @communication(entity=Client; iid=E_C_IID; line=22; sender=E_C_Actor; receiver=M; payload=sign(inv(pk(E_C_Actor)), pair(g(E_C_Actor), pair(loan, P))); channel=regularCh; fact=iknows(sign(inv(pk(E_C_Actor)), pair(g(E_C_Actor), pair(loan, P)))); direction=send)
% @step_label(entity=Client; iid=E_C_IID; line=22; variable=SL; term=2)
step step_002_Client__line_22(A, Amount, B, E_C_Actor, E_C_IID, End_execution_client, Ephemeral_key, M, P, Resp_A, Resp_B) :=
	not(dishonest(E_C_Actor)).
	state_Client(E_C_Actor, E_C_IID, 1, P, M, Amount, Ephemeral_key, Resp_A, Resp_B, End_execution_client, A, B)
	=>
	iknows(sign(inv(pk(E_C_Actor)), pair(g(E_C_Actor), pair(loan, P)))).
	state_Client(E_C_Actor, E_C_IID, 2, P, M, Amount, Ephemeral_key, Resp_A, Resp_B, End_execution_client, A, B)

% @communication(entity=Client; iid=E_C_IID; line=23; sender=M; receiver=E_C_Actor; payload=pair(A_1, B_1); channel=regularCh; fact=iknows(pair(A_1, B_1)); direction=receive)
% @match(entity=Client; iid=E_C_IID; line=23; variable=A; term=A_1)
% @match(entity=Client; iid=E_C_IID; line=23; variable=B; term=B_1)
% @fresh(entity=Client; iid=E_C_IID; line=24; variable=Ephemeral_key; term=Ephemeral_key_1)
% @communication(entity=Client; iid=E_C_IID; line=25; sender=E_C_Actor; receiver=M; payload=pair(crypt(pk(A_1), pair(Amount, pair(E_C_Actor, Ephemeral_key_1))), crypt(pk(B_1), pair(Amount, pair(E_C_Actor, Ephemeral_key_1)))); channel=regularCh; fact=iknows(pair(crypt(pk(A_1), pair(Amount, pair(E_C_Actor, Ephemeral_key_1))), crypt(pk(B_1), pair(Amount, pair(E_C_Actor, Ephemeral_key_1))))); direction=send)
% @step_label(entity=Client; iid=E_C_IID; line=25; variable=SL; term=5)
step step_003_Client__line_23(A, A_1, Amount, B, B_1, E_C_Actor, E_C_IID, End_execution_client, Ephemeral_key, Ephemeral_key_1, M, P, Resp_A, Resp_B) :=
	iknows(pair(A_1, B_1)).
	state_Client(E_C_Actor, E_C_IID, 2, P, M, Amount, Ephemeral_key, Resp_A, Resp_B, End_execution_client, A, B)
	=[exists Ephemeral_key_1]=>
	iknows(pair(crypt(pk(A_1), pair(Amount, pair(E_C_Actor, Ephemeral_key_1))), crypt(pk(B_1), pair(Amount, pair(E_C_Actor, Ephemeral_key_1))))).
	state_Client(E_C_Actor, E_C_IID, 5, P, M, Amount, Ephemeral_key_1, Resp_A, Resp_B, End_execution_client, A_1, B_1)

% @communication(entity=Client; iid=E_C_IID; line=26; sender=M; receiver=E_C_Actor; payload=pair(sign(inv(pk(A)), h(pair(A, pair(Amount, pair(E_C_Actor, Resp_A_1))))), pair(scrypt(Ephemeral_key, Resp_A_1), pair(sign(inv(pk(B)), h(pair(B, pair(Amount, pair(E_C_Actor, Resp_B_1))))), scrypt(Ephemeral_key, Resp_B_1)))); channel=regularCh; fact=iknows(pair(sign(inv(pk(A)), h(pair(A, pair(Amount, pair(E_C_Actor, Resp_A_1))))), pair(scrypt(Ephemeral_key, Resp_A_1), pair(sign(inv(pk(B)), h(pair(B, pair(Amount, pair(E_C_Actor, Resp_B_1))))), scrypt(Ephemeral_key, Resp_B_1))))); direction=receive)
% @match(entity=Client; iid=E_C_IID; line=26; variable=Resp_A; term=Resp_A_2)
% @match(entity=Client; iid=E_C_IID; line=26; variable=Resp_B; term=Resp_B_2)
% @fresh(entity=Client; iid=E_C_IID; line=28; variable=End_execution_client; term=End_execution_client_1)
% @introduce(entity=Client; iid=E_C_IID; line=28; fact=secret(End_execution_client_1, secrecy_End_of_execution, secrecy_End_of_execution_set(IID)))
% @communication(entity=Client; iid=E_C_IID; line=29; sender=E_C_Actor; receiver=P; payload=pair(crypt(pk(P), pair(Amount, pair(E_C_Actor, pair(A, pair(Resp_A_1, pair(B, Resp_B_1)))))), pair(sign(inv(pk(A)), h(pair(A, pair(Amount, pair(E_C_Actor, Resp_A_1))))), pair(sign(inv(pk(B)), h(pair(B, pair(Amount, pair(E_C_Actor, Resp_B_1))))), End_execution_client_1))); channel=regularCh; fact=iknows(pair(crypt(pk(P), pair(Amount, pair(E_C_Actor, pair(A, pair(Resp_A_1, pair(B, Resp_B_1)))))), pair(sign(inv(pk(A)), h(pair(A, pair(Amount, pair(E_C_Actor, Resp_A_1))))), pair(sign(inv(pk(B)), h(pair(B, pair(Amount, pair(E_C_Actor, Resp_B_1))))), End_execution_client_1)))); direction=send)
% @step_label(entity=Client; iid=E_C_IID; line=29; variable=SL; term=8)
step step_004_Client__line_26(A, Amount, B, E_C_Actor, E_C_IID, End_execution_client, End_execution_client_1, Ephemeral_key, IID, M, P, Resp_A, Resp_A_1, Resp_B, Resp_B_1) :=
	child(IID, E_C_IID).
	iknows(pair(sign(inv(pk(A)), h(pair(A, pair(Amount, pair(E_C_Actor, Resp_A_1))))), pair(scrypt(Ephemeral_key, Resp_A_1), pair(sign(inv(pk(B)), h(pair(B, pair(Amount, pair(E_C_Actor, Resp_B_1))))), scrypt(Ephemeral_key, Resp_B_1))))).
	state_Client(E_C_Actor, E_C_IID, 5, P, M, Amount, Ephemeral_key, Resp_A, Resp_B, End_execution_client, A, B).
	not(iknows(Amount)).
	not(iknows(Resp_A)).
	not(iknows(Resp_B))
	=[exists End_execution_client_1]=>
	child(IID, E_C_IID).
	iknows(pair(crypt(pk(P), pair(Amount, pair(E_C_Actor, pair(A, pair(Resp_A_1, pair(B, Resp_B_1)))))), pair(sign(inv(pk(A)), h(pair(A, pair(Amount, pair(E_C_Actor, Resp_A_1))))), pair(sign(inv(pk(B)), h(pair(B, pair(Amount, pair(E_C_Actor, Resp_B_1))))), End_execution_client_1)))).
	secret(End_execution_client_1, secrecy_End_of_execution, secrecy_End_of_execution_set(IID)).
	state_Client(E_C_Actor, E_C_IID, 8, P, M, Amount, Ephemeral_key, Resp_A_1, Resp_B_1, End_execution_client_1, A, B)

% @communication(entity=Clerk; iid=E_C_IID_; line=46; sender=Dummy; receiver=E_C_Actor_; payload=pair(e_request, E_C_M_1); channel=regularCh; fact=iknows(pair(e_request, E_C_M_1)); direction=receive)
% @match(entity=Clerk; iid=E_C_IID_; line=46; variable=M; term=E_C_M_1)
% @communication(entity=Clerk; iid=E_C_IID_; line=47; sender=E_C_Actor_; receiver=E_C_M_1; payload=pair(g(E_C_Actor_), pk(E_C_Actor_)); channel=regularCh; fact=iknows(pair(g(E_C_Actor_), pk(E_C_Actor_))); direction=send)
% @step_label(entity=Clerk; iid=E_C_IID_; line=47; variable=SL; term=3)
step step_005_Clerk__line_46(Client, E_C_Actor_, E_C_Amount, E_C_Ephemeral_key, E_C_IID_, E_C_M, E_C_M_1, Response) :=
	iknows(pair(e_request, E_C_M_1)).
	not(dishonest(E_C_Actor_)).
	not(iknows(g(E_C_Actor_))).
	state_Clerk(E_C_Actor_, E_C_IID_, 1, E_C_M, Client, E_C_Ephemeral_key, E_C_Amount, Response)
	=>
	iknows(pair(g(E_C_Actor_), pk(E_C_Actor_))).
	state_Clerk(E_C_Actor_, E_C_IID_, 3, E_C_M_1, Client, E_C_Ephemeral_key, E_C_Amount, Response)

% @communication(entity=Clerk; iid=E_C_IID_; line=48; sender=E_C_M; receiver=E_C_Actor_; payload=crypt(pk(E_C_Actor_), pair(E_C_Amount_1, pair(Client_1, E_C_Ephemeral_key_1))); channel=regularCh; fact=iknows(crypt(pk(E_C_Actor_), pair(E_C_Amount_1, pair(Client_1, E_C_Ephemeral_key_1)))); direction=receive)
% @match(entity=Clerk; iid=E_C_IID_; line=48; variable=Client; term=Client_1)
% @match(entity=Clerk; iid=E_C_IID_; line=48; variable=Amount; term=E_C_Amount_1)
% @match(entity=Clerk; iid=E_C_IID_; line=48; variable=Ephemeral_key; term=E_C_Ephemeral_key_1)
% @fresh(entity=Clerk; iid=E_C_IID_; line=49; variable=Response; term=Response_1)
% @communication(entity=Clerk; iid=E_C_IID_; line=50; sender=E_C_Actor_; receiver=E_C_M; payload=pair(sign(inv(pk(E_C_Actor_)), h(pair(E_C_Actor_, pair(E_C_Amount_1, pair(Client_1, Response_1))))), scrypt(E_C_Ephemeral_key_1, Response_1)); channel=regularCh; fact=iknows(pair(sign(inv(pk(E_C_Actor_)), h(pair(E_C_Actor_, pair(E_C_Amount_1, pair(Client_1, Response_1))))), scrypt(E_C_Ephemeral_key_1, Response_1))); direction=send)
% @step_label(entity=Clerk; iid=E_C_IID_; line=50; variable=SL; term=6)
step step_006_Clerk__line_48(Client, Client_1, E_C_Actor_, E_C_Amount, E_C_Amount_1, E_C_Ephemeral_key, E_C_Ephemeral_key_1, E_C_IID_, E_C_M, Response, Response_1) :=
	iknows(crypt(pk(E_C_Actor_), pair(E_C_Amount_1, pair(Client_1, E_C_Ephemeral_key_1)))).
	not(iknows(pair(sign(inv(pk(E_C_Actor_)), h(pair(E_C_Actor_, pair(E_C_Amount_1, pair(Client_1, Response_1))))), scrypt(E_C_Ephemeral_key_1, Response_1)))).
	state_Clerk(E_C_Actor_, E_C_IID_, 3, E_C_M, Client, E_C_Ephemeral_key, E_C_Amount, Response)
	=[exists Response_1]=>
	iknows(pair(sign(inv(pk(E_C_Actor_)), h(pair(E_C_Actor_, pair(E_C_Amount_1, pair(Client_1, Response_1))))), scrypt(E_C_Ephemeral_key_1, Response_1))).
	state_Clerk(E_C_Actor_, E_C_IID_, 6, E_C_M, Client_1, E_C_Ephemeral_key_1, E_C_Amount_1, Response_1)

% @communication(entity=Mediator; iid=E_M_IID; line=56; sender=E_M_Actor; receiver=E_M_M; payload=pair(scrypt(g(bob), g(charlie)), pair(scrypt(g(charlie), g(bob)), pair(scrypt(g(bob), g(client)), scrypt(g(client), g(bob))))); channel=regularCh; fact=iknows(pair(scrypt(g(bob), g(charlie)), pair(scrypt(g(charlie), g(bob)), pair(scrypt(g(bob), g(client)), scrypt(g(client), g(bob)))))); direction=send)


section goals:

% @goal(name=secrecy_End_of_execution; line=69; Knowers=Knowers; Msg=Msg)
attack_state secrecy_End_of_execution(Knowers, Msg) :=
	iknows(Msg).
	not(contains(i, Knowers)).
	secret(Msg, secrecy_End_of_execution, Knowers)

% @goal(name=secret_End_of_execution; line=35; Knowers=E_sEoe_Knowers; Msg=E_sEoe_Msg)
attack_state secret_End_of_execution(E_sEoe_Knowers, E_sEoe_Msg) :=
	iknows(E_sEoe_Msg).
	not(contains(i, E_sEoe_Knowers)).
	secret(E_sEoe_Msg, secret_End_of_execution, E_sEoe_Knowers)

\end{lstlisting}

\section{Result of the orchestration problem}

\begin{lstlisting}[breaklines]
INPUT:
	/home/ychevali/Avantssar/spec/aslan/NegativeConstraints.aslan

SUMMARY:
	INCONCLUSIVE

DETAILS:
	NOT_SUPPORTED

BACKEND:
	CL-ATSE 2.5-8_(February_23th_2011)

COMMENTS:
	Failure while reading : sorry; no not<iknows<Amount<5><3622>:text>> allowed in a rule's left-hand side. 

STATISTICS:
	TIME 12 ms
	TESTED 0 transition
	REACHED 0 state
	READING 0.01 seconds
	ANALYSE 0.00 seconds

UNUSED:
	step_001_Environment__line_61
	step_002_Client__line_22
	step_003_Client__line_23

MESSAGES:

\end{lstlisting}

\end{document}